\RequirePackage{fix-cm}
\documentclass[twocolumn]{svjour3}          
\smartqed  
\usepackage{graphicx}
\usepackage{amsmath}
\usepackage[normalem]{ulem}
\newcommand{\stkout}[1]{\ifmmode\text{\sout{\ensuremath{#1}}}\else\sout{#1}\fi}
\newtheorem{Theorem}{Theorem}

\newtheorem{mydef}{Definition}

\usepackage[ruled,vlined,linesnumbered]{algorithm2e}
\usepackage{array, multirow}
\usepackage[scriptsize]{subfigure}
\usepackage[numbers,sort&compress]{natbib}
\usepackage{xcolor}
\usepackage{colortbl}

\let\oldnl\nl
\newcommand{\nonl}{\renewcommand{\nl}{\let\nl\oldnl}}

\SetKwProg{Fn}{Function}{}{}
\SetKw{Continue}{continue}
\begin{document}

\title{Erasing-based lossless compression method for streaming floating-point time series}


\author{Ruiyuan Li$^1$ \and Zheng Li$^1$ \and Yi Wu$^1$ \and Chao Chen$^1$ \and Songtao Guo$^1$ \and Ming Zhang$^2$ \and Yu~Zheng$^{3,4}$}

\authorrunning{Ruiyuan Li et al.} 

\institute{Ruiyuan Li \\ \email{ruiyuan.li@cqu.edu.cn}
           \\
           \\
           Zheng Li \\ \email{zhengli@cqu.edu.cn}
           \\
           \\
           Yi Wu \\ \email{wu\_yi@cqu.edu.cn}
           \\
           \\
           Chao Chen \\ \email{cschaochen@cqu.edu.cn}
           \\
           \\
           Songtao Guo \\ \email{guosongtao@cqu.edu.cn}
           \\
           \\
           Ming Zhang \\ \email{ming.zhang@gzpi.com.cn}
           \\
           \\
           Yu Zheng \\ \email{msyuzheng@outlook.com 
           \\
           \\
           $^1$ Chongqing University, Chongqing, China \\
           $^2$ Guangzhou Urban Planning \& Design Survey Research Institute, Guangzhou, China \\
           $^3$ JD Intelligent Cities Research, Beijing, China \\
           $^4$ Xidian University, Xi'an, China}
}

\date{Received: date / Accepted: date}

\maketitle

\begin{abstract}
There are a prohibitively large number of floating-point time series data generated at an unprecedentedly high rate. An efficient, compact and lossless compression for time series data is of great importance for a wide range of scenarios. Most existing lossless floating-point compression methods are based on the XOR operation, but they do not fully exploit the trailing zeros, which usually results in an unsatisfactory compression ratio. This paper proposes an \underline{E}rasing-based \underline{L}ossless \underline{F}loating-point compression algorithm, i.e., {\em Elf}. The main idea of {\em Elf} is to erase the last few bits (i.e., set them to zero) of floating-point values, so the XORed values are supposed to contain many trailing zeros. The challenges of the erasing-based method are three-fold. First, how to quickly determine the erased bits? Second, how to losslessly recover the original data from the erased ones? Third, how to compactly encode the erased data? Through rigorous mathematical analysis, {\em Elf} can directly determine the erased bits and restore the original values without losing any precision. To further improve the compression ratio, we propose a novel encoding strategy for the XORed values with many trailing zeros. Furthermore, observing the values in a time series usually have similar significand counts, we propose an upgraded version of {\em Elf} named {\em Elf}+ by optimizing the significand count encoding strategy, which improves the compression ratio and reduces the running time further. {Both {\em Elf} and {\em Elf}+} work in a streaming fashion. They take only $\mathcal{O}(N)$ (where $N$ is the length of a time series) in time and $\mathcal{O}(1)$ in space, and achieve a notable compression ratio with a theoretical guarantee. Extensive experiments using 22 datasets show the powerful performance of {\em Elf} and {\em Elf}+ compared with 9 advanced competitors for both double-precision and single-precision floating-point values. Moreover, {\em Elf}+ outperforms {\em Elf} by an average relative compression ratio improvement of 7.6\% and compression time improvement of 20.5\%. 
\end{abstract}

\section{Introduction}\label{sec:introduction}


The advance of sensing devices and Internet of Things \cite{li2015internet, nguyen20216g} has brought about the explosion of time series data. A significant portion of time series data are floating-point values produced at an unprecedentedly high rate in a streaming fashion. For example, there are over ten thousand sensors in a 600,000-kilowatt medium-sized thermal power generating unit, which produce tens of thousands of real-time monitoring floating-point records per second~\cite{zhan2022deepthermal, Yu2021distributed}. Additionally, the sensors on a Boeing 787 can generate up to half a terabyte of data per flight~\cite{jensen2017time}. If these huge floating-point time series data (abbr. time series or time series data in the following) are transmitted and stored in their original format, it would take up a lot of network bandwidth and storage space, which not only causes expensive overhead, but also reduces the system efficiency~\cite{li2020trajmesa, li2021trajmesa} and further affects the usability of some critical applications~\cite{zhan2022deepthermal}.

\begin{figure}[t]
  \centering
  \includegraphics[width=3.2in]{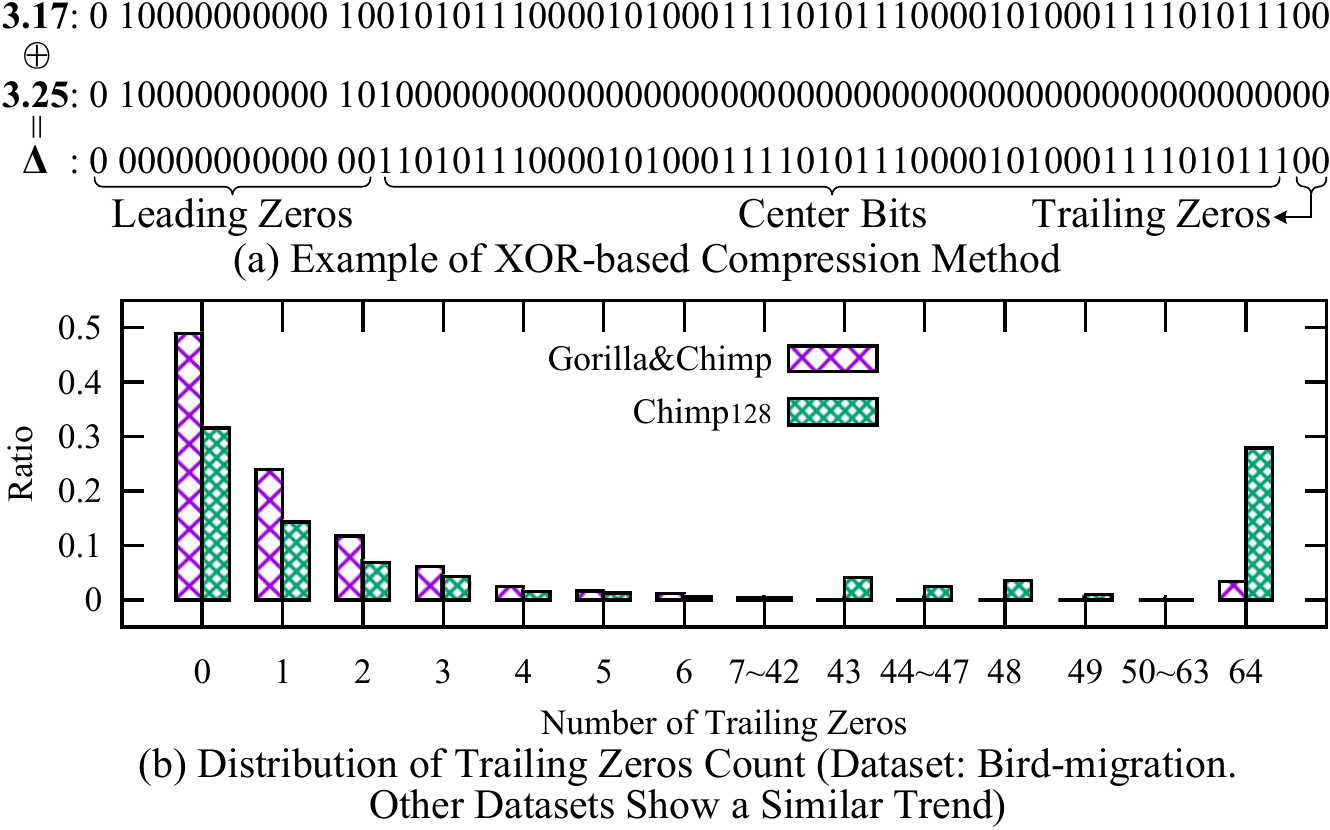}
  \caption{Motivation.}
  \label{fig:Motivation}
\end{figure}


\begin{figure*}[t]
  \centering
  \includegraphics[width=6.88in]{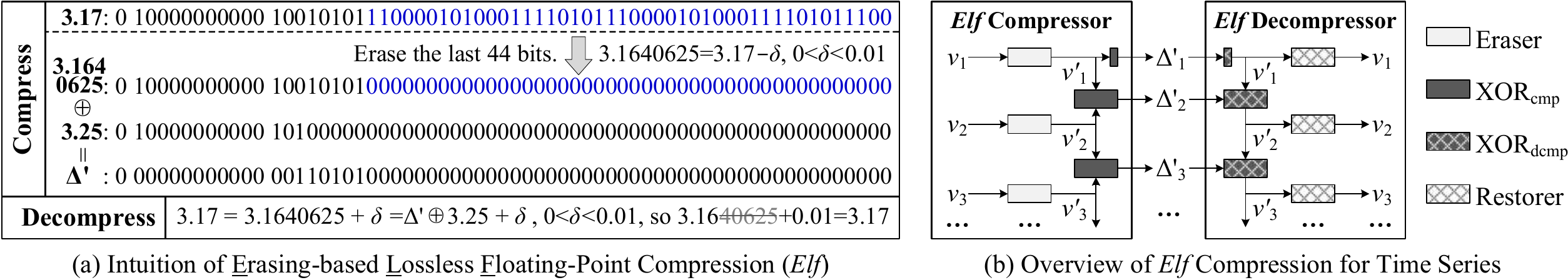}
  \caption{Main Idea of {\em Elf} Compression.}
  \label{fig:MainIdea}
\end{figure*}

One of the best ways is to compress the time series data before transmission and storage. However, it is typically challenging for the compression of floating-point data, because they have a rather complex underlying format~\cite{kahan1996lecture}. General compression algorithms such as LZ4~\cite{collet2013lz4} and Xz~\cite{Xz} do not exploit the intrinsic characteristics (e.g., time ordering) of time-series data. Although they could achieve good compression ratio, they are prohibitively time-consuming. Moreover, most of them run in a batch mode, so they cannot be applied directly to streaming time series data.
%
%
There are two categories of compression methods specifically for floating-point time series data, i.e., lossy compression algorithms and lossless compression algorithms. The former~\cite{lazaridis2003capturing, liang2022sz3, lindstrom2014fixed, zhao2022mdz, zhao2021optimizing, liu2021high, liu2021decomposed} would lose some information, and thus it is not suitable for scientific calculation, data management~\cite{xiao2022time, li2020just, Yu2021distributed, he2022trass, li2022apache} or other critical scenarios~\cite{zhan2022deepthermal}. Imagine the scenes of thermal power generation~\cite{zhan2022deepthermal} and flight~\cite{jensen2017time}, any error could result in disastrous consequences. To this end, lossless floating-point time series compression has attracted extensive interest for decades. One representative lossless algorithm is based on the XOR operation. As shown in Figure~\ref{fig:Motivation}(a), given a time series of double-precision floating-point values, suppose the current value and its previous one are $3.17$ and $3.25$, respectively. If not compressed, each value will occupy 64 bits in its underlying storage (detailed in Section~\ref{subsec:ieee754}). When compressing, the XOR-based compression algorithm performs an XOR operation on $3.17$ and $3.25$, i.e., $\Delta = 3.17 \oplus 3.25$. When decompressing, it recovers $3.17$ through another XOR operation, i.e., $3.17 = \Delta \oplus 3.25$. Because two consecutive values in a time series tend to be similar, the underlying representation of $\Delta$ is supposed to contain many \textbf{leading zeros} (and maybe many \textbf{trailing zeros}). Therefore, we can record $\Delta$ by storing the \textbf{center bits} along with the numbers of leading zeros and trailing zeros, which usually takes up less than 64 bits.

Gorilla~\cite{pelkonen2015gorilla} and Chimp~\cite{liakos2022chimp} are two state-of-the-art XOR-based lossless floating-point compression methods. Gorilla assumes that the XORed result of two consecutive floating-point values is likely to have both many leading zeros and trailing zeros. However, the XORed result actually has very few trailing zeros in most cases. As shown in Figure~\ref{fig:Motivation}(b), if we perform an XOR operation on each value with its previous one (just as Gorilla and Chimp did), there are as many as 95\% XORed results containing no more than 5 trailing zeros. To this end, the work~\cite{liakos2022chimp} proposes Chimp$_{128}$. Instead of using the exactly previous one value, Chimp$_{128}$ selects from the previous 128 values the one that produces an XORed result with the most trailing zeros. As a result, Chimp$_{128}$ can achieve a significant improvement in terms of compression ratio. The lesson we can learn from Chimp$_{128}$ is that, increasing the number of trailing zeros of the XORed results plays a significant role in improving the compression ratio for time series. However, as shown in Figure~\ref{fig:Motivation}(b), when we investigate the trailing zeros' distribution of the XORed results produced by Chimp$_{128}$, there are still up to 60\% of them having no more than 5 trailing zeros. 


This paper proposes an \underline{E}rasing-based \underline{L}ossless \underline{F}loa-ting-point compression algorithm, i.e., {\em Elf}. The intuition of {\em Elf} is simple: if we erase last few bits (i.e., set them to zero) of the floating-point values, we can obtain an XORed result with a large number of trailing zeros. As shown in Figure~\ref{fig:MainIdea}(a), if we erase the last 44 bits of 3.17, we can transform it to 3.1640625. By XORing 3.1640625 with the previous value 3.25 (itself already has a lot of trailing zeros), we can get an XORed result $\Delta'$, which contains as many as 44 trailing zeros (only 2 before erasing as shown in Figure~\ref{fig:Motivation}(a)).

There are three challenges for {\em Elf}. First, how to quickly determine the erased bits? Since there are a prohibitively large number of time series data generated at an unprecedented speed, it requires the erasing step to be as fast as possible. Second, how to losslessly restore the original floating-point data? This paper aims at lossless compression, but the erasing step would introduce some precision loss. It needs a restoring step to recover the original values from the erased ones. Third, how to compactly compress the erased floating point data? Since the distribution of trailing zeros has changed, it calls for a new XOR-based compressor for the erased values.


Figure~\ref{fig:MainIdea}(a) shows the main idea of {\em Elf}. For this example, during the compressing process, we find a small value $\delta$ satisfying $0 < \delta < 0.01$ to erase the bits of 3.17 as many as possible. Therefore, we can obtain an erased value $3.1640625 = 3.17 - \delta$, and encode the XORed result $\Delta' = 3.1640624 \oplus 3.25$ using few bits. During the decompressing process, since we know $3.1640624 = \Delta' \oplus 3.25 = 3.17 - \delta$ and $0 < \delta < 0.01$, we can losslessly recover 3.17 from $\Delta'$ and 3.25 (i.e., $3.16\stkout{40625} + 0.01 = 3.17$). This paper proposes a mathematical method to find $\delta$ in a time complexity of $\mathcal{O}(1)$. Furthermore, we propose a novel XOR-based compressor to encode the XORed results containing many trailing zeros. As shown in Figure~\ref{fig:MainIdea}(b), {\em Elf} consists of Compressor and Decompressor, and works in a streaming fashion. In {\em Elf} Compressor, the original floating-point values $v_i$ flow into {\em Elf} Eraser and are transformed into $v'_i$ with many trailing zeros. Each $v'_i$ (except for $v'_1$) is XORed with its previous value $v'_{i - 1}$. The XORed result $\Delta'_i = v'_i \oplus v'_{i - 1}$ is finally encoded elaborately in {\em Elf} XOR$_{cmp}$. In {\em Elf} Decompressor, each $\Delta'_i$ (except for $\Delta'_1$) is streamed into {\em ELf} XOR$_{dcmp}$ and then XORed with $v'_{i - 1}$. Each $v'_i = \Delta'_i \oplus v'_{i-1}$ is finally fed into {\em Elf} Restorer to get the original value $v_i$. 

\setlength{\tabcolsep}{0.3em} 
\begin{table*}[t]
\caption{Symbols and Their Meanings}
\centering\label{tbl:symbols}
\resizebox{6.88in}{!}{
\begin{tabular}{|c|l|} 
\hline
\textbf{Symbols}			&\textbf{Meanings}	\\
\hline
\hline
$TS = \langle (t_1, v_1), (t_2, v_2), ... \rangle$ & Floating-point time series, where $t_i$ is a timestamp and $v_i$ is a floating-point value\\
\hline
$v$, $v'$				&Original floating-point value, erased floating-point value with long trailing zeros \\
\hline
$DF(v) = \pm(d_{h-1}d_{h-2}...d_0.d_{-1}d_{-2}...d_{l})_{10}$					& Decimal format of $v$, where $d_i \in \{1, 2, ..., 9\}$. ``+'' is usually omitted if $v > 0$ \\
\hline
$BF(v) = \pm(b_{\bar{h}-1}b_{\bar{h}-2}...b_0.b_{-1}b_{-2}...b_{\bar{l}})_2$				& Binary format of $v$, where $b_i \in \{1, 2\}$.  ``+'' is usually omitted if $v > 0$\\
\hline
$DP(v)$, $DS(v)$, $SP(v)$					& Decimal place count, decimal significand count, start decimal significand position of $v$\\
\hline
$s$, $\vec{\boldsymbol{e}} = \langle e_1, e_2, ..., e_{11}\rangle$, $\vec{\boldsymbol{m}} = \langle m_1, m_2, ..., m_{52}\rangle$ & Sign bit, exponent bits, mantissa bits under IEEE 754 format, where $s, e_i, m_j \in \{0, 1\}$\\
\hline
$e$, $\alpha$, $\beta$, $\beta^*$			& Decimal value of $\vec{\boldsymbol{e}}$, alias of $DP(v)$, alias of $DS(v)$, modified $\beta$\\
\hline
\end{tabular}
}
\end{table*}

This paper is extended from our previous work~\cite{li2023elf}. To the best of our knowledge, we provide the first attempt for lossless floating-point compression based on the erasing strategy. In particular, we make the following contributions:

(1)~We propose an erasing-based lossless floating-point compression algorithm named {\em Elf}. {\em Elf} can greatly increase the number of trailing zeros in XORed results by erasing the last few bits, which enhances the compression ratio with a theoretical guarantee.

(2)~Through rigorous theoretical analysis, we can quickly determine the erased bits, and recover the original floating-point values without any precision loss. {\em Elf} takes only $\mathcal{O}(N)$ in time (where $N$ is the length of a time series) and $\mathcal{O}(1)$ in space.

(3)~We also propose an elaborated encoding strategy for the XORed results with many trailing zeros, which further improves the compression performance.

(4)~Observing that most values in a time series have the same significand count, we propose an upgraded version of {\em Elf} called {\em Elf}+ by optimizing the significand count encoding strategy, which further enhance the compression ratio and reduce the compression time.

(5)~We compare {\em Elf} and {\em Elf}+ with 9 state-of-the-art competitors (including 4 floating-point compression algorithms and 5 general compression algorithms) based on 22 datasets. The results show that {\em Elf} and its upgraded version {\em Elf}+ have the best compression ratio among all floating-point compression algorithms in most cases. For example, for double-precision floating-point values, {\em Elf} achieves an average relative compression ratio improvement of 12.4\% over Chimp$_{128}$ and 43.9\% over Gorilla , and {\em Elf}+ further enjoys an average relative improvement of 7.6\% over {\em Elf}. {\em Elf}+ even outperforms most of the compared general compression algorithms, and achieves similar performance to the best general one (i.e., Xz) in terms of compression ratio. However, {\em Elf}+ takes only about 3.86\% compression time and 10.57\% decompression time of Xz.

In the rest of this paper, we give the preliminaries in Section~\ref{sec:preliminary}. In Section~\ref{sec:eraser}, we present the details of {\em Elf} Eraser and Restorer. In Section~\ref{sec:flagOptimization}, we describe the optimized  significand count encoding strategy in {\em Elf}+ Eraser and Restorer. In Section~\ref{sec:xorcompressor}, we elaborate on XOR$_{cmp}$ and XOR$_{dcmp}$. We give some analysis and discussion in Section~\ref{sec:discuss}, and extend the proposed algorithm from double-precision floating-point values to single-precision floating-point values in Section~\ref{sec:singleextend}. The experimental results are shown in Section~\ref{sec:exp}, followed by the related works in Section~\ref{sec:related}. We conclude this paper with future works in Section~\ref{sec:conclude}.

\section{Preliminaries}\label{sec:preliminary}

This section first gives some basic definitions, and then introduces the double-precision floating-point format of IEEE 754 Standard~\cite{kahan1996lecture}. Table~\ref{tbl:symbols} lists the symbols used frequently throughout this paper.

\subsection{Definitions}

\begin{mydef}
\textbf{Floating-Point Time Series.} A floa-ting-point time series $TS = \langle (t_1, v_1), (t_2, v_2), ... \rangle$ is a sequence of pairs ordered by the timestamps in an ascending order, where each pair $(t_i, v_i)$ represents that the floating-point value $v_i$ is recorded in timestamp $t_i$.
\end{mydef}

To compress floating-point time series compactly, one of the best ways is to compress the timestamps and floating-point values separately~\cite{liakos2022chimp, pelkonen2015gorilla, blalock2018sprintz}. For the timestamp compression, existing methods such as delta encoding and delta-of-delta encoding~\cite{pelkonen2015gorilla} can achieve rather good performance, but for the floating-point compression, there is still much room for improvement. To this end, this paper primarily focuses on the compression for floating-point values, particularly for double-precision floating-point values (abbr. \textbf{double values}) in time series (i.e., if not specified, the ``value'' refers to a double value). Single-precision floating-point compression is extended in Section~\ref{sec:singleextend}.

\begin{mydef}
\textbf{Decimal Format and Binary Format.} The decimal format of a double value $v$ is $DF(v) = \pm(d_{h-1}d_{h-2}...d_0.d_{-1}d_{-2}...d_{l})_{10}$, where $d_i \in \{0, 1, ..., 9\}$ for  $l \leq i \leq h-1$, $d_{h-1} \neq 0$ unless $h = 1$, and $d_l \neq 0$ unless $l = -1$. That is, $DF(v)$ would not start with ``0'' except that $h = 1$, and would not end with ``0'' except that $l = -1$. Similarly, the binary format of $v$ is $BF(v) = \pm(b_{\bar{h}-1}b_{\bar{h}-2}...b_0.b_{-1}b_{-2}...b_{\bar{l}})_2$, where $b_j \in \{0, 1\}$ for $\bar{l} \leq j \leq \bar{h}-1$. We have the following relation:
\begin{equation}\label{equ:formatrelation}
	v = \pm \sum_{i = l}^{h-1} d_{i} \times 10^i = \pm \sum_{j = \bar{l}}^{\bar{h}-1} b_{j} \times 2^j
\end{equation}
\end{mydef}

Here, ``$\pm$'' (which means ``$+$'' or ``$-$'') is the sign of $v$. If $v \geq 0$, ``$+$'' is usually omitted. For example, $DF(0) = (0.0)_{10}$, $DF(5.20) = (5.2)_{10}$, and $BF(-3.125)$ $= -(11.001)_2$.

\begin{mydef}\label{def:dpds}
\textbf{Decimal Place Count, Decimal Significand Count and Start Decimal Significand Position}. Given $v$ with its decimal format $DF(v) = \pm(d_{h-1}d_{h-2}...d_0.d_{-1}d_{-2}...d_{l})_{10}$, $DP(v) = |l|$ is called its decimal place count. If for all $l < n \leq i \leq h - 1$, $d_{i} = 0$ but $d_{n - 1} \neq 0$ (i.e., $d_{n-1}$ is the first digit that is not equal to $0$), $SP(v) = n - 1$ is called the start decimal significand position~\footnote{We have $SP(v) = \lfloor log_{10}|v| \rfloor$.}, and $DS(v) = n - l = SP(v) + 1 - l$ is called the decimal significand count. For the case of $v = 0$, we let $DS(v) = 0$ and $SP(v) = undefined$.
\end{mydef}

For example, $DP(3.14) = 2$, $DS(3.14) = 3$, and $SP(3.14) = 0$; $DP(-0.0314) = 4$, $DS(-0.0314) = 3$, and $SP(-0.0314) = -2$; $DP(314.0) = 1$, $DS(314.0) = 4$, and $SP(314.0) = 2$.

\subsection{IEEE 754 Floating-Point Format}\label{subsec:ieee754}

\begin{figure}[t]
  \centering
  \includegraphics[width=3.3in]{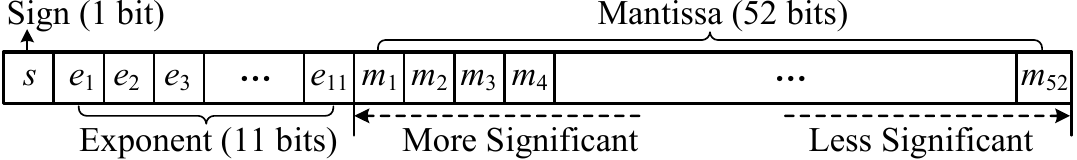}
  \caption{Double-Precision Floating-Point Format.}
  \label{fig:DoubleLayout}
\end{figure}


In accordance with IEEE 754 Standard~\cite{kahan1996lecture}, a double value $v$ is stored with 64 binary bits , where 1 bit is for the sign $s$, 11 bits for the exponent $\vec{\boldsymbol{e}} = \langle e_1, e_2, ..., e_{11}\rangle$, and 52 bits for the mantissa $\vec{\boldsymbol{m}} = \langle m_1, m_2, ..., m_{52}\rangle$, as shown in Figure~\ref{fig:DoubleLayout}. When $v$ is positive, $s = 0$, otherwise $s = 1$. According to the values of $\vec{\boldsymbol{e}}$ and $\vec{\boldsymbol{m}}$, a double value $v$ can be categorized into two main types: \textbf{normal numbers} and \textbf{special numbers}. As normal numbers are the most cases of time series, this paper mainly describes the proposed algorithm for normal numbers. However, our proposed algorithm can be easily extended to special numbers, which will be discussed in Section~\ref{subsec:specialNumbers}.
If $v$ is a normal number (or a normal), its value satisfies:
\begin{equation}\label{equ:normal}
\begin{aligned}
	v &= (-1)^s \times 2^{e - 1023} \times (1.m_1m_2...m_{52})_2\\
	  &= (-1)^s \times 2^{e - 1023} \times (1 + \sum_{i=1}^{52}m_i\times 2^{-i})
\end{aligned}
\end{equation}

\noindent where $e$ is the decimal value of $\vec{\boldsymbol{e}}$~\footnote{We also have $e = \lfloor log_2|v| \rfloor + 1023$.}, i.e., $e = \sum_{i=1}^{11}e_i\times 2^{11-i}$. If let $m_0 = 1$ and $BF(v) = (-1)^s(b_{\bar{h}-1}b_{\bar{h}-2}...b_0.$ $b_{-1}b_{-2}...b_{\bar{l}})_2$, we have:
\begin{equation}\label{equ:bmmap}
b_{-i} = m_{i + e - 1023}, i>0
\end{equation}

As shown in Figure~\ref{fig:DoubleLayout}, in the mantissa $\vec{\boldsymbol{m}} = \langle m_1, m_2,$ $..., m_{52} \rangle$ of a double value $v$, $m_i$ is more significant than $m_j$ for $1 \leq i < j \leq 52$, since $m_i$ contributes more to the value of $v$ than $m_j$.

\section{{\em Elf} Eraser and Restorer}\label{sec:eraser}

In this section, we introduce {\em Elf} Eraser and Restorer since they are strongly correlated.

\subsection{{\em Elf} Eraser}

The main idea of {\em Elf} compression is to erase some less significant mantissa bits (i.e., set them to zeros) of a double value $v$. As a result, $v$ itself and the XORed result of $v$ with its previous value are expected to have many trailing zeros. Note that $v$ and its opposite number $-v$ have the same double-precision floating-point formats except the different values of their signs. That is to say, the compression process for $-v$ can be converted into the one for $v$ if we reverse its sign bit only, and vice versa. To this end, in the rest of the paper, if not specified, we assume $v$ to be \textbf{positive} for the convenience of description. Before introducing the details of {\em Elf} Eraser, we first give the definition of mantissa prefix number.

\begin{mydef}
\textbf{Mantissa Prefix Number.} Given a \\double value $v$ with $\vec{\boldsymbol{m}} = \langle m_1, m_2, ..., m_{52} \rangle$, the double value $v'$ with $\vec{\boldsymbol{m'}} = \langle m'_1, m'_2, ..., m'_{52} \rangle$ is called the mantissa prefix number of $v$ if and only if there exists a number $n \in \{1, 2, ..., 51\}$ such that $m'_i = m_i$ for $1 \leq i \leq n$ and $m'_j = 0$ for $n + 1 \leq j \leq 52$, denoted as $v' = MPN(v, n)$.
\end{mydef}

For example, as shown in Figure~\ref{fig:ExampleOfMPN}, we give four mantissa prefix numbers of 3.17, i.e., $3.17 = MPN(3.17,$ $50)$, $3.169999837875366 = MPN(3.17, 23)$, $3.1640625 = MPN(3.17, 8)$ and $3.125 = MPN(3.17, 4)$.

\begin{figure}[t]
  \centering
  \includegraphics[width=3.3in]{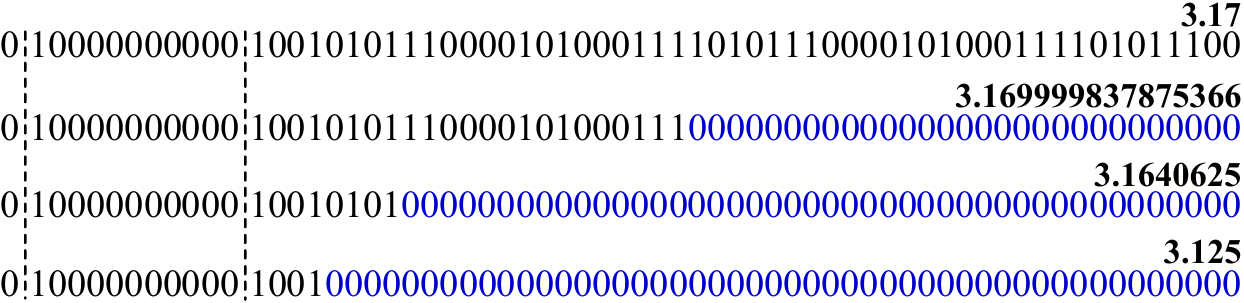}
  \caption{Examples of Mantissa Prefix Number.}
  \label{fig:ExampleOfMPN}
\end{figure}

\subsubsection{Observation}
Our proposed {\em Elf} compression algorithm is based on the following observation: given a double value $v$ with its decimal format $DF(v) = (d_{h-1}d_{h-2}...d_0.d_{-1}d_{-2}...d_l)_{10}$, we can find one of its mantissa prefix numbers $v'$ and a minor double value $\delta$, $0 \leq \delta < 10^l$, such that $v' = v - \delta$. If we retain the information of $v'$ and $\delta$, we can recover $v$ without losing any precision. 

On one hand, there could be many mantissa prefix numbers. Since we aim to maximize the number of trailing zeros of the XORed results, we should select the optimal mantissa prefix number that has the most trailing zeros. Considering the case of $v = 3.17$ shown in Figure~\ref{fig:ExampleOfMPN}, there are many satisfied pairs of $(v', \delta)$, e.g., $(3.17, 0)$, $(3.169999837875366, 0.000000162124634)$ and $(3.1640625, 0.0059375)$. As 3.1640625 has more trailing zeros than 3.169999837875366 and 3.17, the mantissa prefix number 3.1640625 is the most suitable $v'$. 

On the other hand, we find it even unnecessary to figure out and store $\delta$. If $\delta \neq 0$ (we will talk about the case when $\delta = 0$ in Section~\ref{subsubsec:deltazero}) and the decimal place count $DP(v)$ is known, we can easily recover $v$ from $v'$ losslessly. Suppose $\alpha = DP(v)$ and $DF(v') = (d_{h'-1}d_{h'-2}...d_0.d_{-1}d_{-2}...d_{l'})_{10}$, we have~\footnote{Equation~(\ref{equ:recover}) can be implemented by $v = RoundUp(v', \alpha)$, where $RoundUp(v', \alpha)$ is the operation to round $v'$ up to $\alpha$ decimal places.}: 
\begin{equation}\label{equ:recover}
	v = LeaveOut(v', \alpha) + 10^{-\alpha}
\end{equation}

\noindent where $LeaveOut(v', \alpha) = (d_{h'-1}d_{h'-2}...d_0.d_{-1}d_{-2}...d_{-\alpha}\\\stkout{d_{-(\alpha + 1)}...d_{l'}})_{10}$ is the operation that leaves out the digits after $d_{-\alpha}$ in $DF(v')$. For example, given $\alpha = DP(3.17)\\ = 2$ and $v' = 3.1640625$, we have $v = LeaveOut(v', \alpha) + 10^{-\alpha} = (3.16\stkout{40625})_{10} + 10^{-2} = 3.17$.

With the observation above, in the process of compression, what we should do is to find the most appropriate mantissa prefix number $v'$ of $v$ and record $\alpha = DP(v)$. During the decompression process, we can recover $v$ losslessly with the help of $v'$ and $\alpha$ according to Equation~(\ref{equ:recover}). However, there are still two problems left to be addressed. \textbf{Problem~\uppercase\expandafter{\romannumeral1}}: How to find the best mantissa prefix number $v'$ of $v$ with the minimum efforts? \textbf{Problem~\uppercase\expandafter{\romannumeral2}}: How to store the decimal place count $\alpha$ with the minimum storage cost?

\subsubsection{Mantissa Prefix Number Search}\label{subsubsec:mpnsearch}
To address Problem~\uppercase\expandafter{\romannumeral1}, one intuitive idea is to iteratively check all mantissa prefix numbers $v' = MPN(v, i)$ until $\delta = v - v'$ is greater than $10^{-\alpha}$, where $i$ is sequentially from $52$ to $1$. However, this intuitive idea is rather time-consuming since we need to verify the mantissa prefix numbers at most 52 times in the worst case. Although we can enhance the efficiency through a binary search strategy~\cite{bentley1975multidimensional}, the computation complexity $\mathcal{O}(log_252)$ is still high. To this end, we propose a novel mantissa prefix number search method which only takes $\mathcal{O}(1)$.

\begin{Theorem}\label{theorem:f1}
Given a double value $v$ with its decimal place count $DP(v) = \alpha$ and binary format $BF(v) = (b_{\bar{h}-1}b_{\bar{h}-2}...b_0.b_{-1}...b_{\bar{l}})_2$, $\delta = (0.0...0b_{-(f(\alpha) + 1)}b_{-(f(\alpha) + 2)}$ $...b_{\bar{l}})_2$ is smaller than $10^{-\alpha}$, where $f(\alpha) = \lceil |log_210^{-\alpha}|\rceil \\= \lceil \alpha \times log_210\rceil$.
\end{Theorem}
\begin{proof}
$
\begin{aligned}
\delta &= \sum_{i = f(\alpha) + 1}^{|\bar{l}|} b_{-i} \times 2^{-i} \leq \sum_{i = f(\alpha) + 1}^{|\bar{l}|} 2^{-i} < \sum_{i = f(\alpha) + 1}^{+\infty} 2^{-i}\\
&=2^{-f(\alpha)} = 2^{-\lceil \alpha \times log_210\rceil} \leq 2^{- \alpha \times log_210}\\
&=(2^{log_210})^{-\alpha} = 10^{-\alpha}
\end{aligned}
$
\end{proof}

Here, $f(\alpha) = \lceil |log_210^{-\alpha}|\rceil$ means that the decimal value $10^{-\alpha}$ requires exactly $\lceil |log_210^{-\alpha}|\rceil$ binary bits to represent. Suppose $\delta$ is obtained based on Theorem~\ref{theorem:f1}, $v - \delta$ can be regarded as erasing the bits after $b_{-f(\alpha)}$ in $v$'s binary format. Recall that for any $b_{-i}$ in $BF(v)$ where $i > 0$, we can find a corresponding $m_{i + e - 1023}$ according to Equation~(\ref{equ:bmmap}). Consequently, $v-\delta$ can be further deemed as erasing the mantissa bits after $m_{g(\alpha)}$ in $v$'s underlying floating-point format, in which $g(\alpha)$ is defined as:
\begin{equation}\label{equ:g1}
	g(\alpha) = f(\alpha) + e - 1023 = \lceil \alpha \times log_210\rceil + e - 1023
\end{equation}

\noindent where $\alpha = DP(v)$ and $e = (e_1e_2...e_{11})_2 = \sum_{i=1}^{11}e_i\times 2^{11-i}$.

As a result, we can directly calculate the best mantissa prefix number $v'$ by simply erasing the mantissa bits after $m_{g(\alpha)}$ of $v$, which  takes only $\mathcal{O}(1)$.

\subsubsection{Decimal Place Count Calculation}
To solve Problem~\uppercase\expandafter{\romannumeral2}, the basic idea is to utilize $\lceil log_2\alpha_{max} \rceil$ bits for $\alpha$ storage, where $\alpha_{max}$ is the possible maximum value of a decimal place count. According to~\cite{kahan1996lecture}, the minimum value of the double-precision floating-point number is about $4.9 \times 10^{-324}$, so $\alpha_{max} = 324$ and $\lceil log_2\alpha_{max} \rceil = 9$, i.e., the basic method needs as many as 9 bits to store $\alpha$ during the compression process for each double value, which results in a large storage cost and low compression ratio.

Given a double value $v$ with its decimal format $DF(v) \\=(d_{h-1}d_{h-2}...d_0.d_{-1}d_{-2}...d_{l})_{10}$, we notice that its decimal place count $\alpha = DP(v)$ can be calculated by the decimal significand count $\beta = DS(v)$. Since the decimal significand count $\beta$ of a double value would not be greater than 17 under the IEEE 754 Standard~\cite{kahan1996lecture, liakos2022chimp}, it requires much fewer bits to store $\beta$. According to Definition~\ref{def:dpds}, we have $\alpha = DP(v) = |l| = -l$ and $\beta = DS(v) = SP(v) + 1 - l$, so we have:
\begin{equation}\label{equ:beta2alpha}
	\alpha = \beta - (SP(v) + 1)
\end{equation}

Next, we discuss how to get $SP(v)$ without even knowing $v$.

\begin{Theorem}\label{theorem:beta2alpha}
Given a double value $v$ and its best mantissa prefix number $v'$, if $v \neq 10^{-i}$, $i > 0$, then $SP(v) = SP(v')$.
\end{Theorem}
\begin{proof}
Suppose $\alpha = DP(v)$ and $v' = v - \delta$, where $0 \leq \delta < 10^{-\alpha}$.

If $\delta = 0$, i.e., $v = v'$, $DF(v)$ and $DF(v')$ undoubtedly have the same start decimal significand position.

\begin{figure}[t]
  \centering
  \includegraphics[width=3.3in]{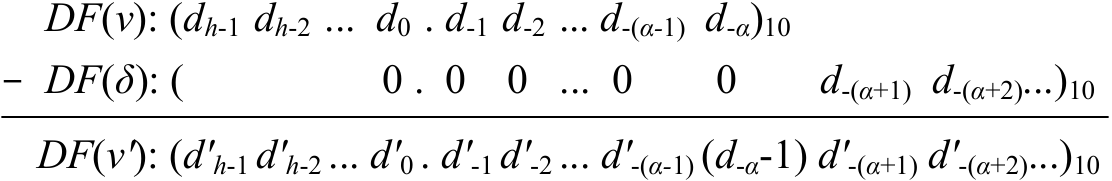}
  \caption{Subtraction in Vertical Form.}
  \label{fig:VerticalFormSubtraction}
\end{figure}

If $\delta \neq 0$, we let $DF(v) = (d_{h-1}d_{h-2}...d_0.d_{-1}...d_{-\alpha})_{10}$, $DF(\delta) = (0.0...0d_{-(\alpha + 1)}d_{-(\alpha + 2)}...)_{10}$ and $DF(v') = (d'_{h-1}d'_{h-2}...d'_0.d'_{-1}d'_{-2}...d'_{-\alpha}...)_{10}$. Figure~\ref{fig:VerticalFormSubtraction} shows the vertical form of the calculation for $v' = v - \delta$, from which we can clearly conclude that $d_i = d'_i$ for $-(\alpha - 1) \leq i \leq h - 1$, and that $d'_{-\alpha} = d_{-\alpha} - 1$. There are two cases: $SP(v) = -\alpha$ and $SP(v) \neq -\alpha$. For the former, we have $d_i = 0$ for $-(\alpha - 1) \leq i \leq h - 1$ and $d_{-\alpha} \neq 0$ according to the definition of the start decimal significand position. Since $v \neq 10^{-i}$, i.e., $d_{-\alpha} \neq 1$, we have $d'_{-\alpha} = d_{-\alpha} - 1 \neq 0$, i.e., $SP(v') = -\alpha = SP(v)$. For the latter, as $v \neq 0$ and $SP(v) \neq -\alpha$, there must exist $j \in \{h-1, h-2, ..., -(\alpha - 1)\}$ such that $d_j \neq 0$. Suppose $d_{j^*}$ is the first one for $d_{j} \neq 0$, i.e., $SP(v) = j^*$. Because $d'_i = d_i$ for $-(\alpha - 1) \leq  i \leq h - 1$, $d'_{j^*}$ is also the first one for $d'_{j} \neq 0$, i.e., $SP(v') = j^* = SP(v)$.
\end{proof}

\begin{figure*}[t]
  \centering
  \includegraphics[width=6.88in]{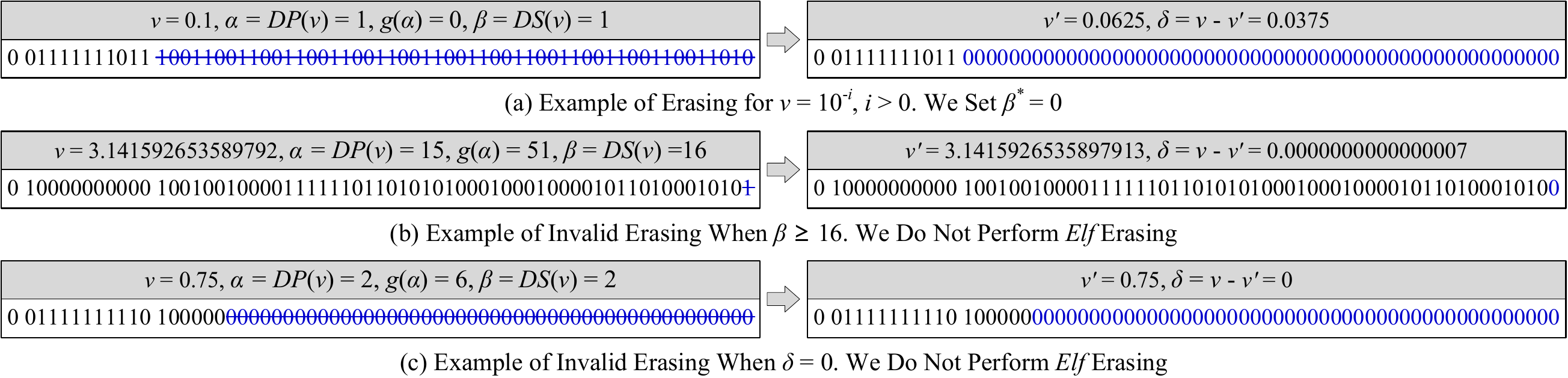}
  \caption{Corner Cases of {\em Elf} Eraser.}
  \label{fig:ExampleOfCornerCase}
\end{figure*}

When $v = 10^{-i}$, $i > 0$, Theorem~\ref{theorem:beta2alpha} does not hold. Figure~\ref{fig:ExampleOfCornerCase}(a) gives an example of $v = 0.1$ with $SP(v) = -1$. If performing the erasing operation on $v$, we get $v'=0.0625$ with $SP(v') = -2$. 

\begin{Theorem}\label{theorem:beta2alpha2}
Given a double value $v = 10^{-i}$, $i > 0$, and its best mantissa prefix number $v'$, we have $SP(v) = SP(v') + 1$.
\end{Theorem}
\begin{proof}
Suppose $\alpha = DP(v)$, we have $\alpha > 0$ and $v = 10^{-\alpha}$. The exponent value of the $v$'s underlying storage is $e = \lfloor log_2|v| \rfloor + 1023 = \lfloor -\alpha \times log_210 \rfloor + 1023$. Based on Equation~(\ref{equ:g1}), we have $g(\alpha) = \lceil \alpha \times log_210 \rceil + \lfloor -\alpha \times log_210 \rfloor = 0$. That is, we will erase all of the mantissa bits, so $v' = (-1)^s \times 2^{\lfloor log_2|v| \rfloor} = 2^{\lfloor log_210^{-\alpha} \rfloor}$. Let $v \div v' = 10^{-\alpha} \div 2^{\lfloor log_210^{-\alpha} \rfloor} = 2^{log_210^{-\alpha}} \div 2^{\lfloor log_210^{-\alpha} \rfloor} = 2^{log_210^{-\alpha} - \lfloor log_210^{-\alpha} \rfloor}$. Since $log_210^{-\alpha} - \lfloor log_210^{-\alpha} \rfloor \in (0, 1)$, we have $v \div v' \in (1, 2)$. Further $v' \in (0.5 \times 10^{-\alpha}, 10^{-\alpha})$. Consequently, $SP(v) = SP(v') + 1$.
\end{proof}

\begin{algorithm}[t]
\small
	\caption{$ElfEraser(v, out)$}\label{alg:elferaser}
	$\alpha \leftarrow DP(v), \beta^{*} \leftarrow DS^*(v)$;\label{alg:elferaser:initial} \tcp{Equation~(\ref{equ:betastar})}
	$\delta \leftarrow$ $\sim\!({\rm 0xffffffffffffffffL}<< (52 - g(\alpha)))\ \&\ v$\;\label{alg:elferaser:delta}
	\If(\tcp*[h]{perform erasing}){$\beta^*<16$ and $\delta\neq0$ and $52-g(\alpha)>4$\label{alg:elferaser:erasecondition}}{
		$out.writeBit(``1"); out.write(\beta^*, 4)$\;\label{alg:elferaser:flagtrue:start}
		$v' \leftarrow ({\rm 0xffffffffffffffffL}<< (52 - g(\alpha)))\ \&\ v$\;\label{alg:elferaser:flagtrue:end}
	}
	\Else(\tcp*[h]{do not perform erasing}){
		$out.writeBit(``0"); v' \leftarrow v$\;\label{alg:elferaser:flagfalse}
	}
	$XOR_{cmp}(v', out)$\;\label{alg:elferaser:xor} 
\end{algorithm}

According to Theorem~\ref{theorem:beta2alpha} and Theorem~\ref{theorem:beta2alpha2}, Equation~(\ref{equ:beta2alpha}) can be rewritten as:
\begin{equation}\label{equ:beta2alpha2}
\alpha = \left\{  
	\begin{array}{ll}
		\beta - (SP(v') + 1) \quad\quad		v \neq 10^{-i}, i>0\\
		\beta - (SP(v') + 2) \quad\quad		v = 10^{-i}, i>0
	\end{array}
\right.
\end{equation}

For any normal number $v$, its decimal significand count $\beta$ will not be zero. Besides, if we know $v = 10^{SP(v)}$, $SP(v) < 0$, we can easily get $v$ from $v'$ by the following equation:
\begin{equation}\label{equ:recover1}
v = 10^{SP(v') + 1}
\end{equation}

\noindent To this end, we can record a modified decimal significand count $\beta^*$ for the calculation of $\alpha$.
\begin{equation}\label{equ:betastar}
\beta^*= DS^*(v)=
\left\{  
	\begin{array}{ll}
		0 				\quad\quad	v = 10^{-i}, i>0\\
		\beta			\quad\quad	others
	\end{array}
\right.
\end{equation}

Although there are 18 possible different values of $\beta^*$, i.e., $\beta^* \in \{0, 1, 2, ..., 17\}$, we do not consider the situations when $\beta^* = 16$ or $17$, because for these two situations, we can only erase a small number of bits but need more bits to record $\beta^*$, which leads to a negative gain (more details will be discussed in Section~\ref{subsec:effectiveness}). For example, as shown in Figure~\ref{fig:ExampleOfCornerCase}(b), given $v = 3.1415926535\\89792$ with $\beta = 16$, we can erase one bit only. In our implementation, we leverage 4 bits to record $\beta^*$ for $0 \leq \beta^* \leq 15$. To ensure a positive gain, when $52 - g(\alpha) \leq 4$, we do not perform the erasing operation.

\subsubsection{When $\delta$ is Zero}\label{subsubsec:deltazero}

As shown in Figure~\ref{fig:ExampleOfCornerCase}(c), given $v = 0.75$, we get $v' = v$ and $\delta = 0$. In this situation, we cannot recover $v$ from $v'$ according to Equation~(\ref{equ:recover}). In fact, $\delta = 0$ indicates that $v$ itself has long trailing zeros. Therefore, once $\delta = 0$, we will keep $v$ as it is.

\subsubsection{Summary of {\em Elf} Eraser}

{\em Elf} Eraser~\cite{li2023elf} utilizes one bit to indicate whether we have erased $v$ or not. As shown in Algorithm~\ref{alg:elferaser}, it takes as input a double value $v$ and an output stream $out$. 

We first calculate the decimal place count $\alpha$ and modified decimal significand count $\beta^*$ based on Equation~(\ref{equ:betastar}), and get $\delta$ by extracting the least $52 - g(\alpha)$ significant mantissa bits of $v$ (Lines~\ref{alg:elferaser:initial}-\ref{alg:elferaser:delta}). 

If the three conditions (i.e., $\beta^* < 16$, $\delta \neq 0$ and $52-g(\alpha) > 4$) hold simultaneously, the output stream $out$ writes one bit of ``1'' to indicate that $v$ should be transformed, followed by 4 bits of $\beta^*$ for the recovery of $v$. We get $v'$ by erasing the least $52 - g(\alpha)$ significant mantissa bits of $v$ (Lines~\ref{alg:elferaser:flagtrue:start}-\ref{alg:elferaser:flagtrue:end}). Otherwise, the output stream $out$ writes one bit of ``0'', and $v'$ is assigned $v$ without any modification (Line~\ref{alg:elferaser:flagfalse}). 

Finally, the obtained $v'$ is passed to an XOR-based compressor together with $out$ for further compression (Line~\ref{alg:elferaser:xor}).

\begin{algorithm}[t]
\small
	\caption{$ElfRestorer(in)$}\label{alg:elfrestorer}
	$flag \leftarrow in.read(1)$\;\label{alg:elfrestorer:readflag}	
	\If(\tcp*[h]{no restoration required}){$flag = 0$}{
		$v \leftarrow XOR_{dcmp}(in)$\;\label{alg:elfrestorer:setsame}
	}
	\Else(\tcp*[h]{perform restoring}){
		$\beta^* \leftarrow in.read(4)$; $v' \leftarrow XOR_{dcmp}(in)$\;\label{alg:elfrestorer:xordcmp}
		\If{$\beta^* = 0$}{
			$v \leftarrow 10^{SP(v') + 1}$;\label{alg:elfrestorer:beta0} \tcp{Equation~(\ref{equ:recover1})}
		}
		\Else{
			$\alpha \leftarrow \beta^* - (SP(v') + 1)$;\label{alg:elfrestorer:recover:start} \tcp{Equation~(\ref{equ:beta2alpha2})}
			$v \leftarrow LeaveOut(v', \alpha) + 10^{-\alpha}$;\label{alg:elfrestorer:recover:end} \tcp{Equation~(\ref{equ:recover})}
		}
	}
	\Return{v}\;\label{alg:elfrestorer:recover:return}
\end{algorithm}

\begin{figure*}[t]
  \centering
  \includegraphics[width=6.0in]{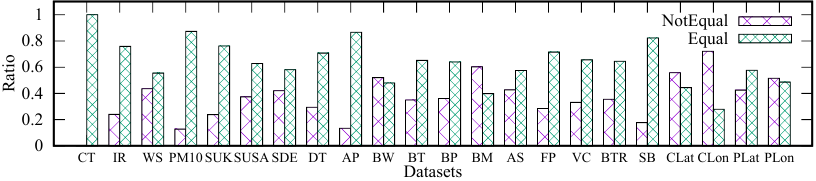}
  \caption{Equal Modified Significant Counts VS Unequal Modified Significant Counts of Two Consecutive Values.}
  \label{fig:BetaStarDistribution}
\end{figure*}

\subsection{{\em Elf} Restorer}

{\em Elf} Restorer is an inverse process of {\em Elf} Eraser. Algorithm~\ref{alg:elfrestorer} depicts the pseudo-code of $Elf$ Restorer~\cite{li2023elf}, which takes in an input stream $in$. First, we read one bit from the input stream $in$ to get the modification flag $flag$ (Line~\ref{alg:elfrestorer:readflag}), which has two cases:

(1)~If $flag$ equals to 0, it means that we have not modified the original value, so we get a value from the XOR-based decompressor and assign it to $v$ directly (Line~\ref{alg:elfrestorer:setsame}). 

(2)~Otherwise, we read 4 bits from $in$ to get the modified decimal significand count $\beta^*$, and then get a value $v'$ from an XOR-based decompressor. If $\beta^*$ equals to 0, $v$ has a format of $10^{-i}$, where $-i = SP(v') + 1$ (Line~\ref{alg:elfrestorer:beta0}). If $\beta^* \neq 0$, we can recover $v$ from $\beta^*$ and $v'$ based on Equation~(\ref{equ:beta2alpha2}) and Equation~(\ref{equ:recover}) (Lines~\ref{alg:elfrestorer:recover:start}-\ref{alg:elfrestorer:recover:end}). 

Finally, the recovered $v$ is returned (Line~\ref{alg:elfrestorer:recover:return}).

\section{{\em Elf}+ Eraser and Restorer}\label{sec:flagOptimization}

In this section, we propose to optimize the significand count encoding strategy, which introduces {\em Elf}+ Eraser and its corresponding Restorer.

\subsection{Observation}

We observe that the values in a time series usually have similar significand counts; therefore, their modified significand counts are also similar (we may interchange the terms of significand count and modified significand count in the following of this section). In Algorithm~\ref{alg:elferaser}, if a value $v$ is to be erased, we always use four bits to record its $\beta^*$, which is not quite effective. One possible method is to record a global $\beta^*_{g}$ of a time series, so the significand count $\beta^*$ of each value $v$ can be represented by $\beta^*_{g}$. However, this method has several drawbacks. First, it requires to know the global significand count before compressing a time series, but this usually cannot be achieved in streaming scenarios. Note that the significand counts of values in a time series are not always the same, so selecting an appropriate $\beta^*_{g}$ is not easy. Second, using $\beta^*_{g}$ to stand for $\beta^*$ might lead to insufficient compression when $\beta^* < \beta^*_g$, or lossy compression when $\beta^* > \beta^*_g$.

To this end, this paper proposes to make the utmost of the modified significand count of the previous one value, which is not only suitable for streaming scenarios and adaptive to dynamic significand counts, but also retains the characteristics of lossless compression. The intuition behind this is that the modified significand count of each value in a time series is likely to be exactly the same as that of the previous value. Figure~\ref{fig:BetaStarDistribution} presents the ratio of equal cases and unequal cases of two consecutive values' modified sigfinicand counts in 22 datasets (for more details please see Section~\ref{sec:exp}) respectively, from which we can see that the equal cases are far more than unequal cases for almost all datasets.

\begin{figure*}[t]
  \centering
  \includegraphics[width=6.8in]{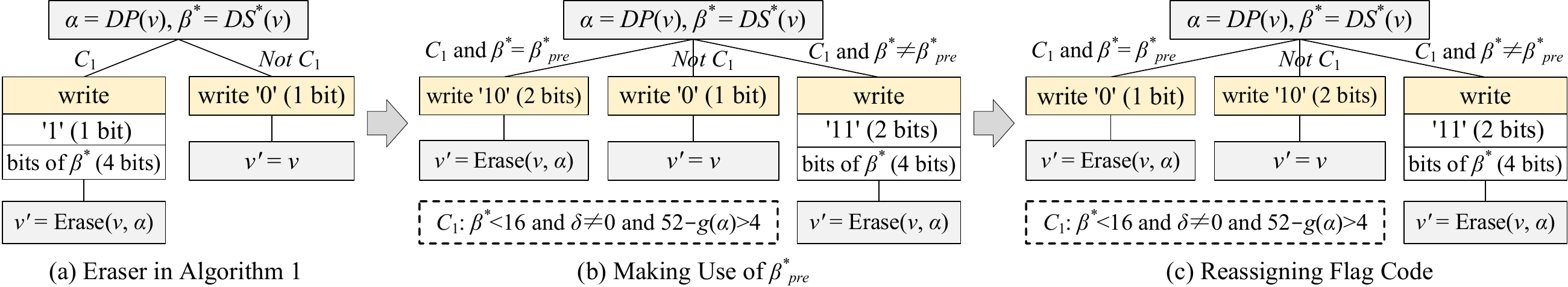}
  \caption{Evolutionary Process of {\em Elf}+ Eraser.}
  \label{fig:ElfEraser}
\end{figure*}

\subsection{{\em Elf}+ Eraser}
We can optimize the modified significand count encoding as follows. If the three conditions (i.e., $C_1$: $\beta^* < 16$, $\delta \neq 0$ and $52-g(\alpha) > 4$) in Algorithm~\ref{alg:elferaser} hold simultaneously, we further check whether the modified significand count $\beta^*$ of the current value is equal to that of the previous one $\beta^*_{pre}$. If $\beta^* = \beta^*_{pre}$, instead of writing the value of $\beta^*$ with 4 bits, we write only one bit of `0', because we can recover $\beta^*$ from $\beta^*_{pre}$, which saves 3 bits. If $\beta^* \neq \beta^*_{pre}$, we would write one more bit of `1' followed by 4 bits of $\beta^*$. As a result, the eraser in Algorithm~\ref{alg:elferaser} (shown in Figure~\ref{fig:ElfEraser}(a)) is converted into the eraser shown in Figure~\ref{fig:ElfEraser}(b). Suppose the ratio of equal cases in a time series is $r_e$. Let $r_e \times 3 - (1 - r_e) > 0$, we have $r_e > 0.25$. That is, if the ratio  of equal cases is greater than 0.25, we can always guarantee a positive gain through the above optimization.

We also notice that the case of ``$C_1$ and $\beta^* = \beta^*_{pre}$'' has the largest proportion among the three cases in Figure~\ref{fig:ElfEraser}(b) for almost all datasets, but we use 2 bits (i.e., `10') to represent this case. According to the coding theory~\cite{huffman1952method}, more frequent cases are encoded with fewer bits. Therefore, we propose to switch the flag codes (i.e., `10' and `0') of case ``$C_1$ and $\beta^* = \beta^*_{pre}$'' and case ``Not $C_1$'' in Figure~\ref{fig:ElfEraser}(b). Finally, the eraser is transformed into the one shown in Figure~\ref{fig:ElfEraser}(c).

\begin{algorithm}[t]
\small
	\caption{$ElfPlusEraser(v, out)$}\label{alg:elferaserplus}
	$\alpha \leftarrow DP(v), \beta^{*} \leftarrow DS^*(v)$; \tcp{Equation~(\ref{equ:betastar})}
	$\delta \leftarrow$ $\sim\!({\rm 0xffffffffffffffffL}<< (52 - g(\alpha)))\ \&\ v$\;\label{alg:elferaserplus:delta}
	\If{$\beta^*<16$ and $\delta\neq0$ and $52-g(\alpha)>4$\label{alg:elferaserplus:erasecondition}}{
		\If{$\beta^* = \beta^*_{pre}$\label{alg:elferaserplus:erase:start}} {
			$out.writeBit(``0")$\;
		}
		\Else{
			$out.writeBit(``11"); out.write(\beta^*, 4)$\;\label{alg:eraserplus:writebeta}
			$\beta^*_{pre} \leftarrow \beta^*$\;\label{alg:elferaserplus:betapre}
		}		
		$v' \leftarrow ({\rm 0xffffffffffffffffL}<< (52 - g(\alpha)))\ \&\ v$\;\label{alg:elferaserplus:erase:end}
	}
	\Else{
		$out.writeBit(``10"); v' \leftarrow v$\;\label{alg:elferaserplus:noterase}
	}
	$XOR_{cmp}(v', out)$\;
\end{algorithm}

Algorithm~\ref{alg:elferaserplus} presents {\em Elf}+ Eraser, which is similar to Algorithm~\ref{alg:elferaser} except two aspects. (1)~We further check if $\beta^* = \beta^*_{pre}$ when $v$ is to be erased (Lines \ref{alg:elferaserplus:erase:start}-\ref{alg:elferaserplus:erase:end}). If $\beta^* = \beta^*_{pre}$, we only write one bit of `0'. Otherwise, we write two bits of `11' and four bits of $\beta^*$. Moreover, we assign $\beta^*$ to $\beta^*_{pre}$ for the compression of the next value (Line~\ref{alg:elferaserplus:betapre}). (2)~The flag codes are different from those in Algorithm~\ref{alg:elferaser}. For example, in Algorithm~\ref{alg:elferaser}, we use one bit of `0' to indicate the case that $v$ would not be erased, but in Algorithm~\ref{alg:elferaserplus} we leverage two bits of `10' for this case (Line~\ref{alg:elferaserplus:noterase}).

\subsection{{\em Elf}+ Restorer}

Correspondingly, {\em Elf} Restorer needs to make some adjustments. As depicted in Algorithm~\ref{alg:elfrestorerplus}, we first read one bit of flag code from the input stream $in$. If the flag code equals to `0', it means that the significand count $\beta^*$ of the current value is the same as that of the previous one, so we set $\beta^*$ as $\beta^*_{pre}$, get the erased value $v'$ from the decompressor, and restore $v$ from $v'$ with the help of $\beta^*$ (Lines~\ref{alg:elfrestoreplus:equal:start}-\ref{alg:elfrestoreplus:equal:end}). If the flag code does not equal to `0', we further read one bit of flag code from $in$. If the new flag code is equal to `0', we just obtain $v$ from $in$ (Line~\ref{alg:elfrestoreplus:noterased}). Otherwise, we get $\beta^*$ by reading four bits from $in$, obtain the erased value $v'$ from the decompressor, and restore $v$ from $v'$ with $\beta^*$ (Lines~\ref{alg:elfrestoreplus:notequal:start}-\ref{alg:elfrestoreplus:notequal:end}). Note that we need also to update $\beta^*_{pre}$ for the decompression of the next value (Line~\ref{alg:elfrestoreplus:notequal:end}). The function of $restore$ (Lines~\ref{alg:elfrestore:func:start}-\ref{alg:elfrestore:func:end}) has the same logic with that in Algorithm~\ref{alg:elfrestorer}.

\begin{algorithm}[t]
\small
	\caption{$ElfPlusRestorer(in)$}\label{alg:elfrestorerplus}
	\If{$in.read(1) = 0$\label{alg:elfrestoreplus:equal:start}}{
		$\beta^* \leftarrow \beta^*_{pre}$; $v' \leftarrow XOR_{dcmp}(in)$\;
		$v \leftarrow restore(\beta^*, v')$\;\label{alg:elfrestoreplus:equal:end}
	}
	\ElseIf{$in.read(1) = 0$}{
		$v \leftarrow XOR_{dcmp}(in)$\;\label{alg:elfrestoreplus:noterased}
	}
	\Else{
		$\beta^* \leftarrow in.read(4)$; $v' \leftarrow XOR_{dcmp}(in)$\;\label{alg:elfrestoreplus:notequal:start}
		$v \leftarrow restore(\beta^*, v')$; $\beta^*_{pre} \leftarrow \beta^*$\;\label{alg:elfrestoreplus:notequal:end}
	}		
	\Return{v}\;
	\Fn{$restore(\beta^*, v')$\label{alg:elfrestore:func:start}}{
		\If{$\beta^* = 0$}{
			$v \leftarrow 10^{SP(v') + 1}$; \tcp{Equation~(\ref{equ:recover1})}
		}
		\Else{
			$\alpha \leftarrow \beta^* - (SP(v') + 1)$; \tcp{Equation~(\ref{equ:beta2alpha2})}
			$v \leftarrow LeaveOut(v', \alpha) + 10^{-\alpha}$; \tcp{Equation~(\ref{equ:recover})}
		}
		\Return{v}\;\label{alg:elfrestore:func:end}
	}
\end{algorithm}

\begin{figure*}[htb]
  \centering
  \includegraphics[width=6.8in]{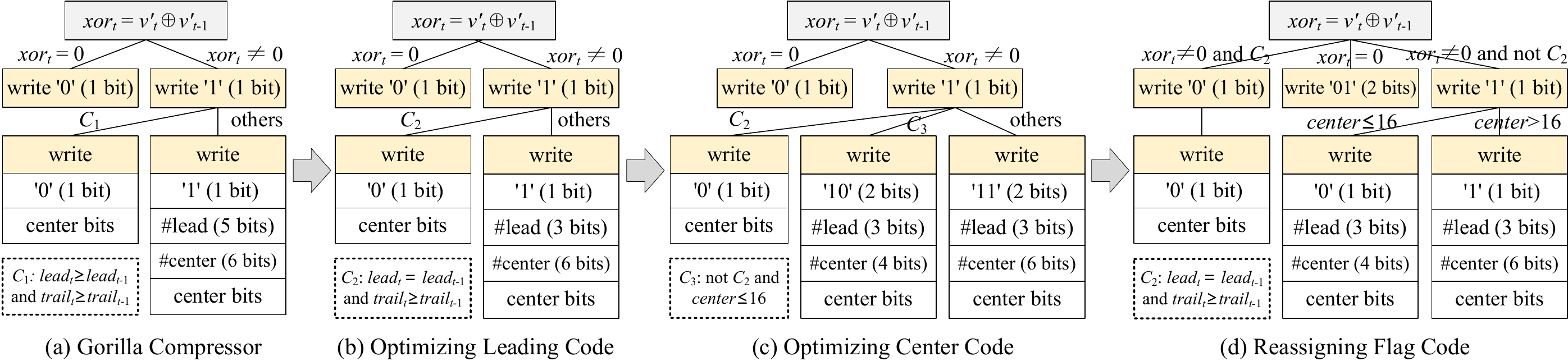}
  \caption{Evolutionary Process of {\em Elf} XOR$_{cmp}$ for $v'_t$ $(t \neq 1)$.}
  \label{fig:XORCompressor}
\end{figure*}

\section{XOR$_{cmp}$ and XOR$_{dcmp}$}\label{sec:xorcompressor}

Theoretically, any existing XOR-based compressor such as Gorilla~\cite{pelkonen2015gorilla} and Chimp~\cite{liakos2022chimp} can be utilized in {\em Elf}. Since the erased value $v'$ tends to contain long trailing zeros, to compress the time series compactly, in this section, we propose a novel XOR-based compressor and the corresponding decompressor. Note that both {\em Elf} and {\em Elf}+ use the same XOR$_{cmp}$ and XOR$_{dcmp}$.

\subsection{{\em Elf} XOR$_{cmp}$}

\subsubsection{First Value Compression}

Existing XOR-based compressors store the first value $v'_1$ of a time series using 64 bits. However, after being erased some insignificant mantissa bits, $v'_1$ tends to have a large number of trailing zeros. As a result, we leverage $\lceil log_265 \rceil = 7$ bits to record the number of trailing zeros $trail$ of $v'_1$ (note that $trail$ can be assigned a total of 65 values from 0 to 64), and store $v'_1$'s non-trailing bits with $64 - trail$ bits. In all, we utilize $71 - trail$ bits to record the first value, which is usually less than 64 bits.

\subsubsection{Other Values Compression}

For each value $v'_t$ that $t > 1$, we store $xor_t = v'_t \oplus v'_{t - 1}$ as most existing XOR-based compressors did. Our proposed XOR-based compressor is extended from Gorilla~\cite{pelkonen2015gorilla} and at the same time borrows some ideas from Chimp~\cite{liakos2022chimp}. 

\textbf{Gorilla Compressor}. As shown in Figure~\ref{fig:XORCompressor}(a), Gorilla compressor checks whether $xor_t$ is equal to 0 or not. If $xor_t = 0$ (i.e., $v'_t = v'_{t-1}$), Gorilla writes one bit of ``0'', and thus it can save many bits without actually storing $v'_t$. If $xor_t \neq 0$, Gorilla writes one bit of ``1'' and further checks whether the condition $C_1$ is satisfied. Here $C_1$ is ``$lead_t \geq lead_{t - 1}$ and $trail_t \geq trail_{t - 1}$'', meaning that the leading zeros count and trailing zeros count of $xor_t$ are greater than or equal to those of $xor_{t-1}$, respectively. If $C_1$ does not hold, after writing a bit of ``1'', Gorilla stores the leading zeros count and center bits count with 5 bits and 6 bits respectively, followed by the actual center bits. Otherwise, $xor_t$ shares the information of leading zeros count and center bits count with $xor_{t - 1}$, which is expected to save some bits.

\textbf{Leading Code Optimization}. Observing that the leading zeros count of an XORed value is rarely more than 30 or less than 8, Chimp~\cite{liakos2022chimp} proposes to use only $log_28 = 3$ bits to represent up to 24 leading zeros. In particular, Chimp leverages 8 exponentially decaying steps (i.e., $0, 8, 12, 16, 18, 20, 22, 24$) to approximately represent the leading zeros count. If the actual leading zeros count is between 0 and 7, Chimp approximates it to be $0$; if it is between 8 and 11, Chimp regards it as $8$; and so on. The condition of $C_1$ is therefore converted into $C_2$, i.e., ``$lead_t = lead_{t - 1}$ and $trail_t \geq trail_{t - 1}$''. By applying this optimization to the Gorilla compressor, we can get a compressor shown in Figure~\ref{fig:XORCompressor}(b).

\textbf{Center Code Optimization}. Both $v'_t$ and $v'_{t - 1}$ are supposed to have many trailing zeros, which results in an XORed value with long trailing zeros. Besides, $v'_t$ would not differentiate much from $v'_{t-1}$ in most cases, contributing to long leading zeros in the XORed value. That is, the XORed value tends to have a small number of center bits (usually not more than $16$). To this end, if the center bits count is less than or equal to 16, we use only $log_216 = 4$ bits to encode it. Although we need one more flag bit, we can usually save one bit in comparison with the original solution. After optimizing the center code, we get a compressor shown in Figure~\ref{fig:XORCompressor}(c).

\textbf{Flag Code Reassignment}. Figure~\ref{fig:XORCompressor}(c) shows that we use only 1 flag bit for the case of $xor_t = 0$, but 2 or 3 flag bits for the cases of $xor_t \neq 0$. As pointed out by Chimp~\cite{liakos2022chimp}, identical consecutive values are not very frequent in floating-point time series. Thus, using only 1 bit to indicate the case of $xor_t = 0$ is not particularly effective. To this end, we reassign the flag codes to the four eases. Therefore, each case uses only 2 bits of flag, as illustrated in Figure~\ref{fig:XORCompressor}(d).

\subsubsection{Summary of {\em Elf} XOR$_{cmp}$}

Algorithm~\ref{alg:elfXORCompressor} depicts the pseudo-code of {\em Elf} XOR$_{cmp}$, which is self-explanatory. In Lines~\ref{alg:elfXORCompressor:first:start}-\ref{alg:elfXORCompressor:first:end}, we deal with the first value of a time series, and in Lines~\ref{alg:elfXORCompressor:other:start}-\ref{alg:elfXORCompressor:other:end}, we handle the four cases shown in Figure~\ref{fig:XORCompressor}(d) respectively. Note that the function $binNumOfLeadingZeros(xor)$ in Line~\ref{alg:elfXORCompressor:binNumOfLead} calculates the approximate leading zeros count of $xor$, as discussed above.

\begin{algorithm}[t]
\small
	\caption{$ElfXOR_{cmp}(v'_t, out)$}\label{alg:elfXORCompressor}
	\If(\tcp*[h]{compress the first value}){$v'_t$ is the first value\label{alg:elfXORCompressor:first:start}}{
		$lead_t \leftarrow \infty$; $trail_t \leftarrow numOfTrailingZeros(v'_t)$\;
		$out.write(trail_t, 7)$\;\label{alg:elfXORCompressor:trail}
		$out.write(nonTrailingBits(v'_t), 64 - trail_t)$\;\label{alg:elfXORCompressor:first:end}
	}
	\Else(\tcp*[h]{compress other values}){
		$xor \leftarrow v'_t \oplus v'_{t-1}$\;\label{alg:elfXORCompressor:other:start}
		\If(\tcp*[h]{case 01}){$xor = 0$}{
			$out.writeBit(``01")$\;
			$lead_t \leftarrow lead_{t - 1}; trail_t \leftarrow trail_{t - 1}$\;
		}
		\Else{
			$lead_t \leftarrow binNumOfLeadingZeros(xor)$\;\label{alg:elfXORCompressor:binNumOfLead}
			$trail_t \leftarrow numOfTrailingZeros(xor)$\;
			$center \leftarrow 64 - lead_t - trail_t$\;
			\If{$lead_t = lead_{t - 1}$ and $trail_t \geq trail_{t-1}$}{
				$out.writeBit(``00")$; \tcp{case 00}
			}
			\ElseIf(\tcp*[h]{case 10}){$center \leq 16$\label{alg:elfXORCompressor:centercountsplit}}{
				$out.writeBit(``10")$\;
				$out.write(lead_t, 3); out.write(center, 4)$\;\label{alg:elfXORCompressor:smallercenter}
			}
			\Else(\tcp*[h]{case 11}){
				$out.writeBit(``11")$\;
				$out.write(lead_t, 3); out.write(center, 6)$\;\label{alg:elfXORCompressor:biggercenter}
			}
			$out.write(centerBits(v'_t), center)$\;\label{alg:elfXORCompressor:other:end}
		}
	}
\end{algorithm}

\subsection{{\em Elf} XOR$_{dcmp}$}

The decompressor takes opposite actions of the compressor. As shown in Algorithm~\ref{alg:elfXORDecompressor}, {\em Elf} XOR$_{dcmp}$ takes an input stream $in$ as input. We decompress the first value in Lines~\ref{alg:elfXORDecompressor:first:start}-\ref{alg:elfXORDecompressor:first:end}, and cope with the four cases respectively in Lines~\ref{alg:elfXORDecompressor:other:start}-\ref{alg:elfXORDecompressor:other:end}. For case 01, the algorithm sets the current value $v'_t$ as the previous one $v'_{t - 1}$. For case 00, case 10 and case 11, we first update the leading zeros count $lead_t$, center bits count $center$ and trailing zeros count $trail_t$ respectively, and then get the current value $v'_t$ (Line~\ref{alg:elfXORDecompressor:other:end}). At last, $v'_t$ is returned to {\em Elf} Restorer (Line~\ref{alg:elfXORDecompressor:return}).

\section{Discussion}\label{sec:discuss}

In this section, we first report the implementation details, and then analyze the effectiveness and complexity of {\em Elf} algorithm. Next, we investigate a possible variant. Finally, we extend {\em Elf} to the special numbers of double values. If not specified, the discussion for {\em Elf} is also applicable to {\em Elf}+.

\subsection{Implementation Details}\label{implementation}

\subsubsection{Significand Count Calculation}

During the implementation, we find that the most time-consuming step of {\em Elf} compression is to calculate the significand counts of floating-point values. Currently, most programming languages do not provide out-of-the-box statements for calculating the significand counts of floating-point values efficiently. The naive method is to first transform a floating-point value into a string, and then calculate its significand count by scanning the string. However, this method runs very slowly since the data type transformation is quite expensive. Other methods, such as {\tt BigDecimal} in Java language, perform even worse as these high-level classes implement many complex but unnecessary logics, which are not suitable for the calculation of significand counts.

\begin{algorithm}[t]
\small
	\caption{$ElfXOR_{dcmp}(in)$}\label{alg:elfXORDecompressor}
	\If(\tcp*[h]{decompress the first value}){it is the first value\label{alg:elfXORDecompressor:first:start}} {
		$lead_t \leftarrow \infty; trail_t \leftarrow in.read(7)$\;\label{alg:elfXORDecompressor:trail}
		$v'_t \leftarrow in.read(64 - trail_t) << trail_t$\;\label{alg:elfXORDecompressor:first:end}
	}
	\Else(\tcp*[h]{decompress other values}){
		$flag \leftarrow in.read(2)$\;\label{alg:elfXORDecompressor:other:start}
		\If(\tcp*[h]{case 01}){$flag = ``01"$}{
			$lead_t \leftarrow lead_{t - 1}; trail_t \leftarrow trail_{t - 1}$\;
			$v'_t \leftarrow v'_{t-1}$\; 
		}
		\Else{
			\If(\tcp*[h]{case 00}\label{alg:elfXORDecompressor:case00:start}){$flag = ``00"$}{
				$lead_t \leftarrow lead_{t - 1}; trail_t \leftarrow trail_{t - 1}$\;
				$center \leftarrow 64 - lead_t - trail_t$;\label{alg:elfXORDecompressor:case00:end}		
			}
			\ElseIf(\tcp*[h]{case 10}){$flag = ``10"$}{
				$lead_t \leftarrow in.read(3); center \leftarrow in.read(4)$\;
				$trail_t \leftarrow 64-lead_t - center$;\label{alg:elfXORDecompressor:smallcenter}
			}
			\Else(\tcp*[h]{case 11}){
				$lead_t \leftarrow in.read(3); center \leftarrow in.read(6)$\;
				$trail_t \leftarrow 64-lead_t - center$;\label{alg:elfXORDecompressor:bigcenter}
			}
			$v'_t \leftarrow (in.read(center) << trail_{t}) \oplus v'_{t-1}$\;\label{alg:elfXORDecompressor:other:end}
		}
	}	
	\Return{$v'_t$;}\label{alg:elfXORDecompressor:return}
\end{algorithm}

\textbf{{\em Elf} Implementation}. We adopt a trial-and-error approach. In particular, for our basic {\em Elf} Eraser (i.e., Algorithm~\ref{alg:elferaser})~\cite{li2023elf}, we iteratively check if the condition ``$v \times 10^i = \lfloor v \times 10^i \rfloor$'' holds (only when the result of $v \times 10^i$ does not have the fractional part, does the condition hold), where $i$ is sequentially from $sp^*$ to at most $sp^* + 17$ (note that the maximum significand count of a double value is 17~\cite{kahan1996lecture, liakos2022chimp}). Here, $sp^*$ is calculated by:
\begin{equation}
sp^* = \left\{  
	\begin{array}{ll}
		1 \quad\quad\quad\quad\quad		SP(v) \geq 0\\
		-SP(v) \quad\quad				SP(v) < 0
	\end{array}
\right.
\end{equation}

\noindent {The value $i$ (denoted as $i^*$) that first makes the equation ``$v \times 10^i = \lfloor v \times 10^i \rfloor$'' hold can be deemed as the decimal place count $\alpha$~\footnote{It is not exactly true for floating-point calculation. We may get an $i^* > DP(v)$. For example, we get 55.00000000000001 for $0.55 \times 10^2$ but $550.0$ for $0.55 \times 10^3$, so $i^* = 3$. However, this will not lead to lossy compression.}. At last, we can get the significand count $\beta = i^* + SP(v) + 1$ according to Equation~(\ref{equ:beta2alpha}).}

\textbf{{\em Elf}+ Implementation}. The verification of the condition ``$v \times 10^i = \lfloor v \times 10^i \rfloor$'' is expected to take $\mathcal{O}(\beta)$ in terms of time complexity. To expedite this process, we take full advantage of the fact that most values in a time series have the same significand count. Particularly, as depicted in Algorithm~\ref{alg:betacal}, we start the verification at $i = max(\beta^*_{pre} - SP(v) - 1, 1)$ based on Equation~(\ref{equ:beta2alpha}). There are two cases. Case 1: $\beta \geq \beta^*_{pre}$. For this case, if ``$v \times 10^i = \lfloor v \times 10^i \rfloor$'' does not hold, we repetitively increase $i$ by 1 until the condition is satisfied (Lines~\ref{alg:betacal:increase:start}-\ref{alg:betacal:increase:end}). Case 2: $\beta < \beta^*_{pre}$. For this case, we should constantly adjust $i$ by decreasing it until the condition ``$i>1$ and $v \times 10^{i-1} = \lfloor v \times 10^{i-1} \rfloor$'' does not hold (Lines~\ref{alg:betacal:decrease:start}-\ref{alg:betacal:decrease:end}). Finally, the significand count is obtained and returned according to Equation~(\ref{equ:beta2alpha}) (Line~\ref{alg:betacal:return}). 

Algorithm~\ref{alg:betacal} is expected to take only $\mathcal{O}(1)$, since the values in a time series have similar significand counts.

\begin{algorithm}[t]
\small
	\caption{$ElfPlusBetaCalculation(v, \beta^*_{pre})$}\label{alg:betacal}
	\textcolor{black}{$i \leftarrow max(\beta^*_{pre} - SP(v) - 1, 1)$; \tcp{Equation~(\ref{equ:beta2alpha})}
	\tcp{Case 1: $\beta \geq \beta^*_{pre}$}
	\While{$v \times 10^i \neq \lfloor v \times 10^i \rfloor$\label{alg:betacal:increase:start}}{
		$i \leftarrow i + 1$\;\label{alg:betacal:increase:end}
	}
	\tcp{Case 2: $\beta < \beta^*_{pre}$}
	\While{$i > 1$ and $v \times 10^{i-1} = \lfloor v \times 10^{i-1} \rfloor$\label{alg:betacal:decrease:start}}{
		$i \leftarrow i - 1$\;\label{alg:betacal:decrease:end}
	}
	\Return{$SP(v) + i + 1$;}\label{alg:betacal:return}\tcp{Equation~(\ref{equ:beta2alpha})}}
\end{algorithm}

\subsubsection{Start Position Calculation}

Another time-consuming operation is calculating the start position $SP(v)$ of a value $v$. In our initial implementation of {\em Elf}~\cite{li2023elf}, we achieve this through $SP(v) = \lfloor log_{10}|v| \rfloor$ directly. However, logarithmic operations are relatively expensive. In {\em Elf}+, we leverage two sorted exponential arrays, i.e., $logArr_1 = \{10^0, 10^1, ..., 10^i, ...\}$ and $logArr_2 = \{10^0 , 10^{-1}, ..., 10^{-j}, ...\}$, to accelerate this process. Particularly, we sequentially scan these two arrays firstly. If $v \geq 1$ and $10^i \leq v < 10^{i+1}$, then $SP(v) = i$; if $v < 1$ and $10^{-j} \leq v < 10^{-(j - 1)}$, then $SP(v) = -j$. In {\em Elf}+ implementation, we set $|logArr_1| = |logArr_2| = 10$, because this can meet the requirements of most time series. If $v \geq 10^{10}$ or $v \leq 10^{-10}$, we call $\lfloor log_{10}|v| \rfloor$ to get $SP(v)$ finally.

We want to emphasize that our {\em Elf} compression algorithm is orthogonal to the ways of significand count calculation and start position calculation. In the future, we may design a special computer instruction or special hardware for these two calculations, which can potentially enhance the efficiency further.

\subsection{Effectiveness Analysis}\label{subsec:effectiveness}

{\em Elf} Eraser transforms a floating-point value to another one with more trailing zeros under a guaranteed bound (see Theorem~\ref{theorem:effectiveness}), so it can potentially improve the compression ratio of most XOR-based compression methods tremendously.

\begin{Theorem}\label{theorem:effectiveness}
Given a double value $v$ with its decimal significand count $\beta = DS(v)$, we can erase $x$ bits in its mantissa, where $51 - \beta log_210 < x < 53 - (\beta - 1)log_210$.
\end{Theorem}
\begin{proof}
Suppose $\alpha = DP(v)$, we have:
$$
DF(v) = \left\{  
	\begin{array}{ll}
		(d_{\beta-\alpha - 1}d_{\beta-\alpha - 2}...d_0.d_{-1}...d_{-\alpha})_{10} \quad {\rm if} v \geq 1\\
		(0.00...d_{\beta-\alpha - 1}d_{\beta-\alpha - 2}...d_{-\alpha})_{10} \quad\quad {\rm if} v < 1
	\end{array}
\right. 
$$\\
$\Longrightarrow 10^{\beta - \alpha - 1} \leq v < 10^{\beta - \alpha}$\\ 
$\Longrightarrow log_210^{\beta - \alpha - 1} \leq log_2v < log_210^{\beta - \alpha}$\\
$\Longrightarrow \lfloor (\beta - \alpha - 1) log_210 \rfloor \leq \lfloor log_2v \rfloor \leq \lfloor (\beta - \alpha) log_210 \rfloor$\\
$\Longrightarrow \lceil \alpha log_210 \rceil + \lfloor (\beta - \alpha - 1) log_210 \rfloor \leq \lceil \alpha log_210 \rceil + \lfloor log_2v \rfloor = g(\alpha) \leq \lceil \alpha log_210 \rceil + \lfloor (\beta - \alpha) log_210 \rfloor$\\
$\Longrightarrow \alpha log_210 + (\beta - \alpha - 1) log_210 - 1 < g(\alpha) < \alpha log_210 + 1 + (\beta - \alpha)log_210$\\
$\Longrightarrow (\beta - 1)log_210 - 1 < g(\alpha) < \beta log_210 + 1$\\
$\Longrightarrow 51 - \beta log_210 < 52 - g(\alpha) = x < 53 - (\beta - 1)log_210$.
\end{proof}

According to Theorem~\ref{theorem:effectiveness}, the number of erased bits is dependent merely on the decimal significand count $\beta$ of the double value. A bigger $\beta$ usually means fewer bits erased. If $\beta \leq 14$, we can erase at least $\lceil 51 - 14 \times log_210 \rceil =  5$ bits, which always guarantees a positive gain. But if $\beta \geq 16$, we can only erase at most $\lfloor 53 - (16 - 1) \times log_210 \rfloor = 3$ bits, leading to a negative gain as it requires at least 4 bits to record $\beta^*$. As a consequence, {\em Elf} compression algorithm keeps $v$ as it is when $DS(v) \geq 16$. {\em Elf}+ usually has a better compression ratio than {\em Elf}, since it uses fewer bits to record $\beta^*$. 

\subsection{Complexity Analysis}

\subsubsection{Time Complexity}

For each value, {\em Elf} Eraser (i.e., Algorithm~\ref{alg:elferaser}) can directly determine the erased bits in $\mathcal{O}(1)$ and perform the erasing operation by efficient bitwise manipulations. In {\em Elf} XOR$_{cmp}$ (i.e., Algorithm~\ref{alg:elfXORCompressor}), all operations can be performed in $\mathcal{O}(1)$.
For {\em Elf} Decompressor, Restorer (i.e., Algorithm~\ref{alg:elfrestorer}) and XOR$_{dcmp}$ (i.e., Algorithm~\ref{alg:elfXORDecompressor}) sequentially read data from an input stream and perform all operations in $\mathcal{O}(1)$. Overall, the time complexity of {\em Elf} is $\mathcal{O}(N)$, where $N$ is the length of a time series.

Our proposed {\em Elf} compression algorithm performs an extra erasing step before actually compressing the data. It is reasonable that the overall computation complexity of {\em Elf} compression algorithm is a little bit higher than that of other XOR-based compression methods, e.g., Gorilla and Chimp.

{\em Elf}+ has the same time complexity with {\em Elf}, but it usually runs faster than {\em Elf}, because it calculates the significand counts of values by making full use of that of the previous one value.

\subsubsection{Space Complexity}

Neither {\em Elf} Eraser nor {\em Elf} Restorer stores any data, while both {\em Elf}+ Eraser and {\em Elf}+ Restorer only record the modified significand count of the previous value. Besides, both XOR$_{cmp}$ and XOR$_{dcmp}$ only store the previous leading zeros count $lead_{t - 1}$, trailing zeros count $trail_{t - 1}$ and value $v'_{t - 1}$. To this end, the space complexity of {\em Elf} and {\em Elf}+ is both $\mathcal{O}(1)$.

\subsection{A Possible Variant Discussion}

In the erasing process, we let $v' = v - \delta$ where $0 \leq \delta < 10^{-\alpha}$. Can we let $0 \leq \delta < k \times 10^{-\alpha}$, $k \in \{1, 2, ..., 9\}$, which is supposed to make $v'$ have more trailing zeros?

The decimal value $k \times 10^{-\alpha}$ can be represented by $f_k(\alpha) = \lceil |log_2(k \times 10^{-\alpha})| \rceil = \lceil |log_2k -\alpha log_210| \rceil$ binary bits. Since $k < 10$ and $\alpha \geq 1$, $f_k(\alpha) = \lceil \alpha log_210 - log_2k \rceil$. According to Theorem~\ref{theorem:f1}, $\delta = \sum_{i = f_k(\alpha) + 1}^{|\bar{l}|}b_i \times 2^{-i} \leq \sum_{i = f_k(\alpha) + 1}^{|\bar{l}|} 2^{-i} < \sum_{i = f_k(\alpha) + 1}^{+\infty} 2^{-i} = 2^{-f_k(\alpha)} = 2^{-\lceil \alpha log_210 - log_2k \rceil} \leq 2^{-(\alpha log_210 - log_2k)} = 2^{log_2(k \times 10^{-\alpha})} = k \times 10^{-\alpha}$. That is to say, if we erase the bits after $b_{-f_k(\alpha)}$ in $BF(v)$, we can still recover $v$ by $LeaveOut(v', \alpha) + k'\times 10^{-\alpha}$, where $LeaveOut$ has the same meaning with that in Equation~(\ref{equ:recover}), and $k' \in \{1, 2, ..., k\}$. But it requires $\lceil log_2k \rceil$ bits to store $k'$. We call this method {\em Elf$_k$}.

\begin{Theorem}\label{theorem:elfk}
$Elf_k$ will not achieve a better gain than $Elf$.
\end{Theorem}
\begin{proof}
Suppose $y$ is the additional number of bits that $Elf_k$ can erase over $Elf$ (i.e., $Elf_1$), then $y - \lceil log_2k \rceil$ is the gain of $Elf_k$ over $Elf$. We have:
$y = (52 - g_k(\alpha)) - (52 - g_1(\alpha)) = \lceil \alpha log_210 \rceil - \lceil \alpha log_210 - log_2k \rceil$
$\Longrightarrow \alpha log_210 - (\alpha log_210 - log_2k + 1) < y < (\alpha log_210 + 1) - (\alpha log_210 - log_2k) \Longrightarrow log_2k - 1 < y < log_2k + 1 \Longrightarrow log_2k - 1 - \lceil log_2k \rceil < y - \lceil log_2k \rceil < log_2k + 1 - \lceil log_2k \rceil \Longrightarrow -2 < y - \lceil log_2k \rceil < 1$. It means that $Elf_k$ would consume the same bits with or one more bit than $Elf$. 
\end{proof}


\subsection{{\em Elf} for Special Numbers}\label{subsec:specialNumbers}

As shown in Figure~\ref{fig:DoubleLayout}, according to the values of $\vec{\boldsymbol{e}}$ and $\vec{\boldsymbol{m}}$, there are four types of special numbers:

(1)~\textbf{Zero}. If $\forall i \in \{1, 2, ..., 11\}$, $e_i = 0$ and $\forall j \in \{1, 2, ..., 52\}$, $m_j = 0$, then $v$ represents a zero. 

(2)~\textbf{Infinity}. If $\forall i \in \{1, 2, ..., 11\}$, $e_i = 1$ and $\forall j \in \{1, 2, ..., 52\}$, $m_j = 0$, then $v$ stands for an infinity.

(3)~\textbf{Not a Number}. If $\forall i \in \{1, 2, ..., 11\}$, $e_i = 1$ and $\exists j \in \{1, 2, ..., 52\}$, $m_j = 1$, then $v$ is not a number (i.e., $v = NaN$).

(4)~\textbf{Subnormal Number}. If $\forall i \in \{1, 2, ..., 11\}$, $e_i = 0$ and $\exists j \in \{1, 2, ..., 52\}$, $m_j = 1$, then $v$ is a subnormal number (or a subnormal). In this case, we have the following equation:
\begin{equation}\label{equ:subnormal}
\begin{aligned}
	v &= (-1)^s \times 2^{-1022} \times (0.m_1m_2...m_{52})_2\\
	  &= (-1)^s \times 2^{-1022} \times \sum_{i=1}^{52}m_i\times 2^{-i}
\end{aligned}
\end{equation}

For these four special numbers, their restorers, compressors and decompressors are the same with that of normal numbers, but their erasers need to be tailored carefully.

\textbf{Zero and Infinity Eraser}. If $v$ is a zero or infinity, we do not perform {\em Elf} erasing because all its mantissa bits are already 0s.

\textbf{NaN Eraser}. If $v$ is NaN, in order to make its trailing zeros as many as possible, we perform the $NaN_{norm}$ operation on it, which sets $m_1 = 1$ and $m_i = 0$ for $i \in \{2, 3, ..., 52\}$, i.e., 
\begin{equation}\label{equ:nannorm}
v' = NaN_{norm}(v) = {\rm 0xfff8000000000000L}\ \&\ v
\end{equation}

\textbf{Subnormal Number Eraser}. According to Equation~(\ref{equ:normal}) and  Equation~(\ref{equ:subnormal}), subnormal numbers can be regarded as the special cases of normal numbers by setting $e = 1$ and $m_0 = 0$. As a result, we can compress subnormal numbers in the same way of normal numbers using {\em Elf} Eraser.

%
%

\section{\textcolor{black}{Extension to Single Values}}\label{sec:singleextend}

\textcolor{black}{A single-precision floating-point value (abbr. \textbf{single} \\ \textbf{value}) has a similar underlying storage layout to that of a double value, but it takes up only 32 bits, where 1 bit is for the sign, 8 bits for the exponent, and 23 bits for the mantissa. To this end, when applying {\em Elf} to single values, we should make the following modifications.}

\textbf{Modifications for equations}. We change ``1023'' in Equation~(\ref{equ:normal}), Equation~(\ref{equ:bmmap}) and Equation~(\ref{equ:g1}) to ``127'', and ``52'' in Equation~(\ref{equ:normal}) and Equation~(\ref{equ:subnormal}) to ``23''. We should also change ``1022'' in Equation~(\ref{equ:subnormal}) to ``126''. For Equation~(\ref{equ:nannorm}), we let $NaN_{norm}(v) = {\rm 0xffc00000} \& v$.

\textbf{Modifications for Eraser and Restorer}. First, we change ``$\delta \leftarrow$ $\sim\!({\rm 0xffffffffffffffffL} << (52 - g(\alpha)))\ \&\ v$'' (i.e., Line~\ref{alg:elferaser:delta} in Algorithm~\ref{alg:elferaser} and Line~\ref{alg:elferaserplus:delta} in Algorithm~\ref{alg:elferaserplus}) to ``$\delta \leftarrow$ $\sim\!({\rm 0xffffffff}<< (23 - g(\alpha)))\ \&\ v$''. Similarly, we change ``$v' \leftarrow ({\rm 0xffffffffffffffffL} << (52 - g(\alpha)))\ \&\ v$'' (i.e., Line~\ref{alg:elferaser:flagtrue:end} in Algorithm~\ref{alg:elferaser} and Line~\ref{alg:elferaserplus:erase:end} in Algorithm~\ref{alg:elferaserplus}) to ``$v' \leftarrow ({\rm 0xffffffff}<< (23 - g(\alpha)))\ \&\ v$''.

Second, since the maximum significand count of a single value is 7~\cite{kahan1996lecture, liakos2022chimp}, we need only $\lceil log_27 \rceil = 3$ bits to store $\beta^*$. Consequently, we change ``$out.write(\beta^*, 4)$'' (i.e., Line~\ref{alg:elferaser:flagtrue:start} in Algorithm~\ref{alg:elferaser} and Line~\ref{alg:eraserplus:writebeta} in Algorithm~\ref{alg:elferaserplus}) to ``$out.write(\beta^*, 3)$''. Correspondingly, we change ``$\beta^* \\ \leftarrow in.read(4)$'' (i.e., Line~\ref{alg:elfrestorer:xordcmp} in Algorithm~\ref{alg:elfrestorer} and Line~\ref{alg:elfrestoreplus:notequal:start} in Algorithm~\ref{alg:elfrestorerplus}) to ``$\beta^* \leftarrow in.read(3)$''.

Third, for single values, the erasing condition ``$\beta^*<16$ and $\delta\neq0$ and $52 - g(\alpha)>4$'' (i.e., Line~\ref{alg:elferaser:erasecondition} in Algorithm~\ref{alg:elferaser} and Line~\ref{alg:elferaserplus:erasecondition} in Algorithm~\ref{alg:elferaserplus}) should be converted into ``$\beta^* < 8$ and $\delta \neq 0$ and $23 - g(\alpha)>3$''. Here, $\beta^*$ will always be less than $2^3 = 8$, so the condition ``$\beta^* < 8$'' can be omitted.

\textbf{Modifications for XOR$_{cmp}$ and XOR$_{dcmp}$}. For the first place, as a single value occupies only 32 bits, we should change all ``64'' in Algorithm~\ref{alg:elfXORCompressor} and Algorithm~\ref{alg:elfXORDecompressor} into ``32'' for single values.

Second, the number of trailing zeros of a single value would not be greater than 32, so we can use only $\lceil log_2(32 \\+ 1) \rceil = 6$ bits to record $trail_t$ in Line~\ref{alg:elfXORCompressor:trail} of Algorithm~\ref{alg:elfXORCompressor}. Similarly, in Line~\ref{alg:elfXORDecompressor:trail} of Algorithm~\ref{alg:elfXORDecompressor}, we only read 6 bits from $in$ to obtain $trail_t$.

Third, in {\em Elf} for double values, we leverage 3 bits for 8 exponentially decaying steps (i.e., 0, 8, 12, 16, 18, 20, 22, 24) to approximately represent the leading zeros count, and 4 or 6 bits to store the number of center bits. As a single value takes up only 32 bits, the leading zeros count and center bits count of two consecutive single values will be much less than that of two consecutive double values, respectively. To this end, in {\em Elf} for single values, although we still utilize 3 bits to approximately represent the leading zeros count, the exponentially decaying steps would be 0, 6, 10, 12, 14, 16, 18 and 20 (corresponding to Line~\ref{alg:elfXORCompressor:binNumOfLead} in Algorithm~\ref{alg:elfXORCompressor}), which provides a fine-grained representation of leading zeros count. Furthermore, for single values, after the erasing and XORing operations, the center bits count of an XORed value is likely to be less than 8. In view of that, if the center bits count is less than 8, we use only $log_28 = 3$ bits to encode it (corresponding to Line~\ref{alg:elfXORCompressor:centercountsplit} and Line~\ref{alg:elfXORCompressor:smallercenter} in Algorithm~\ref{alg:elfXORCompressor}, and Line~\ref{alg:elfXORDecompressor:smallcenter} in Algorithm~\ref{alg:elfXORDecompressor}); otherwise, we use $log_232 = 5$ bits (corresponding to Line~\ref{alg:elfXORCompressor:biggercenter} in Algorithm~\ref{alg:elfXORCompressor} and Line~\ref{alg:elfXORDecompressor:bigcenter} in Algorithm~\ref{alg:elfXORDecompressor}).

\section{Experiments}\label{sec:exp}

\subsection{Datasets and Experimental Settings}

\subsubsection{Datasets}

To verify the performance of {\em Elf} compression algorithm, we adopt 22 datasets including 14 time series and 8 non time series, which are further divided into three categories respectively according to their average decimal significand counts (as described in Table~\ref{tbl:datasets}). Apart from the datasets used by Chimp~\cite{liakos2022chimp}, we also add three datasets (i.e., Vehicle-charge, City-lat and City-lon) to enrich the non time series with small and medium decimal significand counts. Each time series is ordered by the timestamps, while each non time series is in a random order given by its data publisher.

\setlength{\tabcolsep}{0.33em} 
\begin{table}[t]
\caption{Details of Datasets}
\centering\label{tbl:datasets}
\resizebox{3.33in}{!}{
\begin{tabular}{|c|c|c|c|c|c|} 
\hline
\multicolumn{3}{|c|}{\textbf{Dataset}}	&\textbf{\#Records}	&\textbf{$\beta$}	&\textbf{Time Span}\\
\hline
\hline
\multirow{14}{*}{\rotatebox[origin=c]{90}{\textbf{Time Series}}}	& \multirow{4}{*}{\textbf{Small $\beta$}}	&City-temp (CT)			&2,905,887			&3						&25 years\\
										&&IR-bio-temp (IR)		&380,817,839		&3						&7 years\\
										&&Wind-speed (WS)			&199,570,396	&2						&6 years\\
										&&PM10-dust	(PM10)		&222,911			&3						&5 years\\
\cline{2-6}
										&&Stocks-UK (SUK)			&115,146,731		&5						&1 year\\
										&&Stocks-USA (SUSA)			&374,428,996		&4						&1 year\\
										&&Stocks-DE	(SDE)		&45,403,710			&6						&1 year\\
										&\textbf{Medium}&Dewpoint-temp (DT)		&5,413,914			&4						&3 years\\
										&$\beta$&Air-pressure (AP)		&137,721,453		&7					&6 years\\
										&&Basel-wind (BW)			&124,079			&8						&14 years\\
										&&Basel-temp (BT)			&124,079			&9						&14 years\\
										&&Bitcoin-price	(BP)	&2,741				&9						&1 month\\
										&&Bird-migration (BM)		&17,964				&7						&1 year\\
\cline{2-6}
										&\textbf{Large $\beta$}&Air-sensor (AS)			&8,664				&17						&1 hour\\
\hline
\hline
\multirow{8}{*}{\rotatebox[origin=c]{90}{\textbf{Non Time Series}}}&\multirow{2}{*}{\textbf{Small $\beta$}}&Food-price (FP)			&2,050,638			&3						&-\\
										&&Vehicle-charge (VC)		&3,395				&3						&-\\
\cline{2-6}
										&&Blockchain-tr (BTR)		&231,031			&5						&-\\
										&\textbf{Medium}&SD-bench (SB)			&8,927				&4		&-\\
										&$\beta$&City-lat (CLat)			&41,001				&6						&-\\
										&&City-lon (CLon)			&41,001				&7						&-\\
\cline{2-6}
										&\multirow{2}{*}{\textbf{Large $\beta$}}&POI-lat (PLat)				&424,205			&16						&-\\
										&&POI-lon (PLon)				&424,205			&16						&-\\
\hline
\end{tabular}
}
\end{table}

\textbf{City-temp}~\cite{CityTemp}, collected by the University of Dayton to record the temperature of major cities around the world.

\textbf{IR-bio-temp}~\cite{IRBioTemp}, which exhibits the changes in the temperature of infrared organisms.

\textbf{Wind-speed}~\cite{WindDir}, which describes the wind speed.

\textbf{PM10-dust}~\cite{PM10Dust}, which records near real-time measurements of PM10 in the atmosphere.

\textbf{Stocks-UK, Stocks-USA and Stocks-DE}~\cite{Stocks}, which contain the stock exchange prices of UK, USA and German respectively.

\textbf{Dewpoint-temp}~\cite{DewpointTemp}, which records relative dew point temperature observed by sensors floating on rivers and lakes.

\textbf{Air-pressure}~\cite{AirPressure}, which shows Barometric pressure corrected to sea level and surface level.

\textbf{Basel-wind and Basel-temp}~\cite{Basel}, which respectively record the historical wind speed and temperature of Basel, Switzerland.

\textbf{Bitcoin-price}~\cite{influxdb2data}, which includes the price of Bitcoin in dollar exchange rate.

\textbf{Bird-migration}~\cite{influxdb2data}, an online dataset of animal tracking data that records the position of birds and the vegetation.

\textbf{Air-sensor}~\cite{influxdb2data}, a synthetic dataset recording air sensor data with random noise.

\textbf{Food-price}~\cite{WorldFoodPrice}, global food prices data from the World Food Programme.

\textbf{Vehicle-charge}~\cite{VehicleCharge}, which records the total energy use and charge time of a collection of electric vehicles.

\textbf{Blockchain-tr}~\cite{BlockchianTr}, which records the transaction value of Bitcoin for a single day.

\textbf{SD-bench}~\cite{SSD}, which describes the performance of multiple storage drives through a standardized series of tests.

\textbf{City-lat, City-lon}~\cite{citylat}, which records the latitude and longitude of the cities and towns all over the world.

\textbf{POI-lat, POI-lon}~\cite{POI}, the coordinates in radian of Position-of-Interests (POI) extracted from Wikipedia.

\subsubsection{Baselines}

We compare {\em Elf} compression algorithm with \textbf{four} state-of-the-art lossless floating-point compression methods (i.e., Gorilla~\cite{pelkonen2015gorilla}, Chimp~\cite{liakos2022chimp}, Chimp$_{128}$~\cite{liakos2022chimp} and FPC \cite{burtscher2007high}) and \textbf{five} widely-used general compression methods (i.e., Xz~\cite{Xz}, Brotli~\cite{alakuijala2018brotli}, LZ4~\cite{collet2013lz4}, Zstd~\cite{collet2016zstd} and Snappy~\cite{snappy}). The initial implementation~\cite{li2023elf} of the proposed method is termed as {\em Elf}, and the one that adopts significand count optimization and start position optimization is termed as {\em Elf}+. By regarding {\em Elf} Eraser (or {\em Elf}+Eraser) as a preprocessing step, we also compare \textbf{three} variants of Gorilla, Chimp and Chimp$_{128}$, denoted as Gorilla+Eraser, Chimp+Eraser and Chimp$_{128}$ +Eraser (or Gorilla+Eraser$^+$, Chimp+Eraser$^+$ and \\Chimp$_{128}$+Eraser$^+$) respectively, to verify the effectiveness of the erasing and XOR$_{cmp}$ strategies. Most implementations of these competitors are extended from~\cite{liakos2022chimp}. To make a fair comparison, we optimize the stream implementation of Gorilla as the same as Chimp~\cite{liakos2022chimp}, which improves the efficiency of Gorilla tremendously.  All source codes and datasets are publicly available~\cite{Elf}.

\subsubsection{Metrics}

We verify the performance of various methods in terms of three metrics: compression ratio, compression time and decompression time. Note that the compression ratio is defined as the ratio of the compressed data size to the original one.

\subsubsection{Settings} As Chimp~\cite{liakos2022chimp} did, we regard 1,000 records of each dataset as a block. Each compression method is executed on up to 100 blocks per dataset, and the average metrics of one block are finally reported. By default, we regard each value as a double value.
All experiments are conducted on a personal computer equipped with Windows 11, 11th Gen Intel(R) Core(TM) i5-11400 @ 2.60GHz CPU and 16GB memory. The JDK (Java Development Kit) version is 1.8.

\subsection{Overall Comparison for Double Values}

\setlength{\tabcolsep}{0.05em} 
\begin{table*}[t]
\caption{Overall comparison with baselines for double values (the best values in each group are marked in \textbf{bold}). The compression ratio, compression time and decompression time are the average measurements on one block (i.e., 1,000 values).}
\centering\label{tbl:overallcomparison}
\addtolength{\leftskip} {-2cm}
\addtolength{\rightskip}{-2cm}
\resizebox{6.880in}{!}{
\begin{tabular}{|c|c|c||c|c|c|c|c|c|c|c|c|c|c|c|c|c||c||c|c|c|c|c|c|c|c||c|} 
\hline
\multicolumn{3}{|c||}{\multirow{3}{*}{\textbf{Dataset}}} & \multicolumn{15}{c||}{\textbf{Time Series}} & \multicolumn{9}{c|}{\textbf{Non Time Series}} \\
\cline{4-27}

\multicolumn{3}{|c||}{}& \multicolumn{4}{c|} {Small $\beta$} & \multicolumn{9}{c|}{Medium $\beta$} & {\scriptsize Large $\beta$} & \multirow{2}{*}{Avg.}	& \multicolumn{2}{c|}{Small $\beta$} 	& \multicolumn{4}{c|}{Medium $\beta$} & \multicolumn{2}{c||}{Large $\beta$}&\multirow{2}{*}{Avg.}\\

\cline{4-17}\cline{19-26}

\multicolumn{3}{|c||}{}& CT	& IR	& WS	& PM10	& SUK	& SUSA	& SDE	& DT	& AP	& BW	& BT	& BP	& BM	& AS & 	& FP	& VC	& BTR	& SB	& CLat	& CLon	& PLat	& PLon	& \\
\hline
\hline

\multirow{11}{*}{\rotatebox[origin=c]{90}{\textbf{Compression Ratio}}}	&  	\multirow{5}{*}{\rotatebox[origin=c]{90}{\textbf{Floating}}}	& Gorilla	&0.85	&0.64	&0.83	&0.48	&0.58	&0.68	&0.72	&0.83	&0.73	&0.99	&0.94	&0.84	&0.79	&0.82	&0.76	&0.58	&1.00	&0.74	&0.63	&1.03	&1.03	&1.03	&1.03	&0.88\\
	& 		& Chimp	&0.64	&0.59	&0.81	&0.46	&0.52	&0.64	&0.67	&0.77	&0.65	&0.88	&0.85	&0.77	&0.72	&\textbf{0.77}	&0.70	&0.47	&0.86	&0.67	&0.55	&0.92	&0.98	&\textbf{0.90}	&\textbf{0.99}	&0.79\\
	&		& Chimp$_{128}$	&0.32	&0.24	&{0.23}	&0.21	&0.29	&{0.23}	&0.27	&0.35	&0.54	&0.71	&\textbf{0.47}	&0.72	&0.50	&\textbf{0.77}	&0.42	&0.34	&0.36	&0.55	&{0.27}	&0.78	&0.85	&\textbf{0.90}	&\textbf{0.99}	&0.63\\
	&		& FPC	&0.75	&0.61	&0.85	&0.50	&0.74	&0.70	&0.73	&0.82	&0.67	&0.92	&0.90	&0.81	&0.75	&0.82	&0.75	&0.62	&0.91	&0.69	&0.59	&0.96	&1.00	&0.95	&1.00	&0.84\\
	&		& Elf	&{0.25}	&{0.21}	&0.25	&{0.16}	&{0.22}	&0.24	&{0.26}	&{0.31}	&{0.31}	&{0.59}	&0.58	&{0.56}	&{0.42}	&0.85	&{0.37}	&{0.23}	&{0.34}	&{0.36}	&{0.27}	&{0.56}	&{0.63}	&0.96	&1.06	&{0.55}\\
	&		& Elf+&	\textbf{0.22}&	\textbf{0.15}&	\textbf{0.20}&	\textbf{0.11}&	\textbf{0.19}&	\textbf{0.18}&	\textbf{0.23}&	\textbf{0.26}&	\textbf{0.25}&	\textbf{0.56}&	0.52&	\textbf{0.50}&	\textbf{0.38}&	0.86&	\textbf{0.33}&	\textbf{0.22}&	\textbf{0.29}&	\textbf{0.30}&	\textbf{0.23}&	\textbf{0.51}&	\textbf{0.60}&	0.98&	1.07&	\textbf{0.52}\\
\cline{2-27}

& \multirow{5}{*}{\rotatebox[origin=c]{90}{\textbf{General}}} & Xz	&\textbf{0.18}	&\textbf{0.16}	&\textbf{0.15}	&\textbf{0.11}	&\textbf{0.16}	&\textbf{0.17}	&\textbf{0.19}	&\textbf{0.27}	&\textbf{0.47}	&\textbf{0.57}	&\textbf{0.35}	&\textbf{0.63}	&\textbf{0.43}	&\textbf{0.79}	&\textbf{0.33}	&\textbf{0.23}	&\textbf{0.23}	&\textbf{0.40}	&\textbf{0.13}	&\textbf{0.60}	&\textbf{0.63}	&\textbf{0.93}	&\textbf{0.96}	&\textbf{0.51}\\
	&		& Brotli	&0.20	&0.18	&0.17	&0.12	&0.19	&0.20	&0.22	&0.32	&0.51	&0.61	&0.39	&0.71	&0.47	&0.85	&0.37	&0.26	&0.28	&0.43	&0.14	&0.65	&0.68	&0.94	&\textbf{0.96}	&0.54\\
	&		& LZ4	&0.36	&0.36	&0.37	&0.27	&0.39	&0.39	&0.41	&0.52	&0.69	&0.69	&0.54	&0.87	&0.61	&1.01	&0.53	&0.41	&0.47	&0.53	&0.30	&0.79	&0.82	&1.00	&1.00	&0.67\\
	&		& Zstd	&0.22	&0.24	&0.19	&0.14	&0.22	&0.24	&0.26	&0.38	&0.58	&0.61	&0.41	&0.75	&0.51	&0.91	&0.40	&0.30	&0.34	&0.45	&0.17	&0.68	&0.71	&0.94	&\textbf{0.96}	&0.57\\
	&		& Snappy	&0.29	&0.30	&0.27	&0.21	&0.32	&0.32	&0.35	&0.51	&0.73	&0.75	&0.54	&0.99	&0.61	&1.00	&0.51	&0.39	&0.42	&0.54	&0.25	&0.83	&0.87	&1.00	&1.00	&0.66\\
\hline
\hline

\multirow{11}{*}{\rotatebox[origin=c]{90}{\textbf{Compression Time ($\mu$s)}}}	&  	\multirow{5}{*}{\rotatebox[origin=c]{90}{\textbf{Floating}}}	& Gorilla	&\textbf{18}	&\textbf{21}	&\textbf{17}	&\textbf{15}	&\textbf{17}	&\textbf{17}	&\textbf{17}	&\textbf{18}	&\textbf{20}	&\textbf{21}	&\textbf{20}	&\textbf{19}	&\textbf{18}	&\textbf{20}	&\textbf{18}	&\textbf{16}	&\textbf{19}	&\textbf{18}	&\textbf{16}	&\textbf{19}	&\textbf{19}	&\textbf{19}	&\textbf{19}	&\textbf{18}\\
	& 		& Chimp	&23	&\textbf{21}	&22	&18	&23	&22	&23	&24	&\textbf{20}	&26	&25	&24	&25	&27	&23	&21	&24	&22	&20	&26	&26	&23	&26	&23\\
	&		& Chimp$_{128}$	&23	&23	&22	&20	&24	&22	&25	&26	&38	&47	&35	&48	&38	&50	&32	&27	&27	&39	&23	&48	&48	&45	&46	&38\\
	&		& FPC	&34	&40	&40	&40	&28	&28	&28	&31	&40	&42	&47	&27	&30	&38	&35	&39	&43	&43	&41	&42	&48	&40	&48	&43\\
	&		& Elf	&51	&53	&59	&50	&54	&56	&58	&57	&51	&73	&69	&63	&65	&87	&60	&52	&55	&62	&48	&64	&70	&71	&72	&62\\	
	&		& Elf+	&34	&35	&53	&30 &40 &39	&43	&39	&59	&72 &54	&42	&51	&82	&48	&41	&42	&43	&35	&51	&63	&48	&66	&49\\
\cline{2-27}

& \multirow{5}{*}{\rotatebox[origin=c]{90}{\textbf{General}}} & Xz	&948	&1106	&810	&1056	&877	&836	&900	&1045	&1959	&1527	&1100	&1531	&1444	&2146	&1235	&898	&1636	&1036	&1040	&1252	&1516	&1476	&1351	&1276\\
	&		& Brotli	&1639	&1685	&1557	&1449	&1584	&1611	&1693	&1702	&2074	&1792	&1715	&1729	&1827	&1798	&1704	&1741	&1674	&1755	&1522	&1692	&1712	&1628	&1633	&1669\\
	&		& LZ4	&1082	&1106	&963	&984	&966	&976	&952	&1091	&1285	&1013	&1010	&1001	&1000	&1026	&1032	&985	&974	&1060	&976	&988	&986	&966	&957	&987\\
	&		& Zstd	&209	&\textbf{212}	&112	&\textbf{208}	&177	&112	&\textbf{117}	&218	&317	&259	&291	&271	&\textbf{256}	&277	&217	&211	&\textbf{227}	&251	&202	&236	&245	&206	&\textbf{113}	&211\\
	&		& Snappy	&\textbf{195}	&236	&\textbf{52}	&214	&\textbf{169}	&\textbf{56}	&172	&\textbf{195}	&\textbf{179}	&\textbf{189}	&\textbf{200}	&\textbf{169}	&261	&\textbf{158}	&\textbf{175}	&\textbf{188}	&250	&\textbf{190}	&\textbf{200}	&\textbf{207}	&\textbf{238}	&\textbf{178}	&149	&\textbf{200}\\
\hline
\hline

\multirow{11}{*}{\rotatebox[origin=c]{90}{\textbf{\scriptsize Decompression Time ($\mu$s)}}}	&  	\multirow{5}{*}{\rotatebox[origin=c]{90}{\textbf{Floating}}}	& Gorilla	&\textbf{16}	&18	&17	&21	&\textbf{16}	&17	&\textbf{17}	&\textbf{17}	&\textbf{18}	&\textbf{23}	&\textbf{18}	&\textbf{16}	&\textbf{17}	&\textbf{20}	&\textbf{18}	&\textbf{16}	&\textbf{18}	&\textbf{17}	&\textbf{16}	&\textbf{17}	&\textbf{17}	&\textbf{17}	&\textbf{17}	&\textbf{17}\\
	& 		& Chimp	&24	&22	&24	&19	&22	&24	&24	&54	&19	&30	&26	&27	&25	&25	&26	&21	&26	&24	&21	&26	&26	&24	&26	&24\\
	&		& Chimp$_{128}$	&17	&\textbf{16}	&\textbf{16}	&\textbf{15}	&18	&\textbf{16}	&18	&18	&22	&28	&21	&26	&22	&25	&20	&18	&19	&22	&17	&26	&26	&23	&24	&22\\
	&		& FPC	&28	&28	&26	&29	&25	&24	&25	&25	&32	&27	&31	&24	&26	&34	&28	&28	&29	&29	&29	&30	&36	&28	&35	&31\\
	&		& Elf	&38	&44	&46	&43	&37	&45	&44	&45	&41	&58	&53	&48	&48	&29	&44	&33	&44	&49	&39	&52	&57	&31	&33	&42\\
	&		& Elf+	&27	&28	&33	&27	&28	&29 &31	&30 &44	&41	&45	&34	&36	&35	&33	&30	&33	&33	&30	&41	&49 &33	&36	&36\\
\cline{2-27}

& \multirow{5}{*}{\rotatebox[origin=c]{90}{\textbf{General}}} & Xz	&161	&147	&114	&125	&156	&133	&148	&226	&435	&427	&284	&479	&345	&629	&272	&196	&194	&312	&126	&434	&461	&664	&663	&381\\
	&		& Brotli	&61	&58	&36	&53	&41	&43	&69	&70	&109	&97	&79	&93	&87	&100	&71	&103	&70	&86	&58	&243	&85	&86	&77	&101\\
	&		& LZ4	&40	&\textbf{35}	&\textbf{18}	&\textbf{37}	&\textbf{19}	&\textbf{19}	&\textbf{18}	&42	&56	&42	&\textbf{38}	&\textbf{40}	&\textbf{38}	&\textbf{44}	&\textbf{35}	&\textbf{36}	&\textbf{37}	&\textbf{39}	&37	&\textbf{38}	&\textbf{37}	&35	&\textbf{19}	&\textbf{35}\\
	&		& Zstd	&46	&48	&30	&42	&31	&31	&50	&45	&99	&66	&113	&72	&62	&68	&57	&45	&47	&60	&44	&47	&48	&43	&32	&46\\
	&		& Snappy	&\textbf{38}	&54	&20	&38	&\textbf{19}	&21	&20	&\textbf{39}	&\textbf{49}	&\textbf{40}	&42	&41	&46	&48	&37	&40	&39	&\textbf{39}	&\textbf{36}	&42	&\textbf{37}	&\textbf{32}	&43	&38\\
\hline

\end{tabular}
}
\end{table*}

Table~\ref{tbl:overallcomparison} shows the performance of different compression algorithms on all datasets. We group the datasets into two categories (i.e., Time Series and Non Time Series), and investigate the performance of floating-point compression algorithms and general compression algorithms on each group of datasets, respectively.

\subsubsection{Compression Ratio} 

With regard to the compression ratio, we have the following observations from Table~\ref{tbl:overallcomparison}.

(1) \textbf{{\em Elf} VS floating-point compression algorithms}. Among all the floating-point compression algorithms, {\em Elf} has the best compression ratio on almost all datasets (excluding {\em Elf}+). In particular, for the time series datasets, compared with Gorilla and FPC, {\em Elf} has an average relative improvement of $(0.76 - 0.37) / 0.76 \approx 51\%$. Chimp has optimized the coding of Gorilla, and its upgraded version Chimp$_{128}$ resorts to a hash table (up to 33KB memory occupation) for fast searching an appropriate value in previous 128 data records. Therefore, they can achieve a significant improvement over Gorilla. However, thanks to the erasing technique and elaborate XOR$_{cmp}$, {\em Elf} can still achieve relative improvement of 47\% and 12\% over Chimp and Chimp$_{128}$ respectively on the time series datasets. Note that {\em Elf} has a lower memory footprint (i.e., $\mathcal{O}(1)$) in comparison with Chimp$_{128}$. For the non time series datasets, {\em Elf} is also relatively $(0.63 - 0.55) / 0.63 \approx 12.7\%$ better than the best competitor Chimp$_{128}$. We notice that there are few datasets that Chimp$_{128}$ is slightly better than {\em Elf} in terms of compression ratio. For the datasets of WS, SUSA and BT, we find that there are many duplicate values within 128 consecutive records. In this case, Chimp$_{128}$ can use only 9 bits to represent the same value. For the datasets of AS, PLat and PLon, since they have large decimal significand counts, {\em Elf} does not perform erasing but still consumes some flag bits. As pointed out by~\cite{liakos2022chimp}, real-world floating point measurements often have a decimal place count of one or two, which usually results in small or medium $\beta$. To this end, {\em Elf} can achieve good performance in most real-world scenarios.

(2) \textbf{{\em Elf} VS general compression algorithms}. Most of the general compression algorithms have a good compression ratio. However, upon most occasions, {\em Elf} is still better than LZ4, Zstd and Snappy (with average relative improvement of 30.2\%, 7.5\% and 27.5\% respectively for the time series datasets, and 18\%, 3.5\% and 16.7\% respectively for the non time series datasets), and shows a similar performance to Xz and Brotli in terms of compression ratio. Moreover, in comparison with non time series datasets, {\em Elf} can achieve more improvement over general compression algorithms for time series datasets (e.g., 30.2\% v.s. 18\% for LZ4). It is because non time series datasets do not have a time-based ordering, which reduces the usefulness of exploiting previous values.

(3) \textbf{Different decimal significand counts}. As shown in Table~\ref{tbl:overallcomparison}, with a larger $\beta$, both general and floating-point compression algorithms suffer from a lower compression ratio, since a larger $\beta$ means a more complex data layout. To this end, the poor compression ratio on datasets with a large $\beta$ is not just a problem for {\em Elf}. It is a common and interesting problem worthy of further exploration.

\textcolor{black}{(4) \textbf{{\em Elf}+ VS {\em Elf}}. Table~\ref{tbl:overallcomparison} shows that for both time series and non-time series with small and medium $\beta$, {\em Elf}+ always performs better than {\em Elf} with regard to compression ratio. This is because {\em Elf}+ takes full advantage of the fact that most values in a time series have the same significand count, and thus it encodes $\beta^*$ with fewer bits. Thanks to this optimization, {\em Elf}+ even outperforms the best competitor Chimp$_{128}$ for datasets WS and SUSA, in which Chimp$_{128}$ has a slightly better compression ratio than {\em Elf}. On the contrary, for values with big $\beta$, {\em Elf}+ performs a bit worse than {\em Elf}, since {\em Elf}+  utilizes two bits to indicate the case of not erasing, while {\em Elf} only takes up one bit for this case.}

\subsubsection{Compression Time and Decompression Time}

As shown in the lower parts of Table~\ref{tbl:overallcomparison}, we have the following observations.

(1) The general compression algorithms take one or two orders of magnitude of more compression time than floating-point compression algorithms on average. For example, although Xz can achieve a slightly better compression ratio than {\em Elf}, it takes as much as 200 times longer than {\em Elf}. Even for the fastest general compression algorithms Zstd and Snappy, they still take about 3 times longer than {\em Elf}, which prevents them from being applied to real-time scenarios.

(2) {\em Elf} takes a little more time than other floating-point compression algorithms during both compression and decompression processes. Compared with other floa-ting-point compression algorithms, {\em Elf} adds an erasing step and a restoring step, which inevitably takes more time. However, the difference is not obvious, since they are all on the same order of magnitude. Gorilla has the least compression time and decompression time, because it considers fewer cases (see Figure~\ref{fig:XORCompressor}(a)) compared with Chimp and Chimp$_{128}$.

(3) Compared with compression time, the distinction of decompression time among different algorithms (except for Xz) is insignificant, since most algorithms sequentially read the decompression stream directly. As a result, most algorithms focus more on the trade-off between compression ratio and compression time.

\textcolor{black}{(4) For almost all datasets, {\em Elf}+ takes less time than {\em Elf} during both compression and decompression processes. For example, on average, {\em Elf}+ takes about 79.5\% of the compression time of {\em Elf}, and this ratio turns into 80.2\% for decompression time. These improvements owe to two reasons. First, when compressing a value, {\em Elf}+ leverages the significand count of its previous value, which avoids iteratively trying to get the decimal place count from scratch. Second, in the processes of compression and decompression, to get the start position $SP(v)$, {\em Elf}+ adopts more efficient numerical checks instead of expensive logarithmic operations. We also notice that for values with larger $\beta$, the efficiency improvement of {\em Elf}+ is not so significant (sometimes it is even slightly worse than {\em Elf} due to experimental errors). This is because if $\beta^* \geq 16$, {\em Elf}+ will not store $\beta^*$; therefore, the optimization of significand counts will not take effect.}

\subsubsection{Summary} In summary, {\em Elf} can usually achieve remarkable compression ratio improvement for both time series datasets and non time series datasets, with the affordable cost of more time. \textcolor{black}{{Furthermore, \em Elf}+ performs better than {\em Elf} in terms both of compression ratio and running time.}

One interesting question is how much efficiency gain can we benefit from {\em Elf} \textcolor{black}{or {\em Elf}+} over the best competitor, i.e., $Chimp_{128}$? Consider a scenario of data transmission. Suppose the raw data size is $D$, the compression ratio is $\eta$, and the rates of compression, decompression and transmission are $r_{cmp}$, $r_{dcmp}$ and $r_{tr}$, respectively. The latency of the whole data from sending to receiving is: $t = D/r_{cmp} + D/r_{dcmp} + D \times \eta /r_{tr}$. According to Table~\ref{tbl:overallcomparison}, in terms of the average metrics for time series, we have $r_{cmp}^{Elf} = 1000 \times 64 / (60 \times 10^{-6}) \approx 1.07 \times 10^9$ bits/s, $r_{dcmp}^{Elf} = 1000 \times 64 / (44 \times 10^{-6}) \approx 1.45 \times 10^9$ bits/s, and $\eta^{Elf} = 0.37$. Similarly, $r_{cmp}^{Chimp_{128}} = 2 \times 10^9$ bits/s, $r_{dcmp}^{Chimp_{128}} = 3.2 \times 10^{9}$ bits/s, and $\eta^{Chimp_{128}} = 0.42$. Therefore, $t^{Elf} / t^{Chimp_{128}} \approx  (1.62 + 0.37 \times 10^9 / r_{tr}) / \\ (0.81 + 0.42 \times 10^9 / r_{tr})$, where $r_{tr}^{Elf} = r_{tr}^{Chimp_{128}} = r_{tr}$.  Let $t^{Elf} / t^{Chimp_{128}} < 1$, we have $r_{tr} < 6.17 \times 10^7$ bits/s. That is, when the transmission rate is smaller than $ 6.17 \times 10^7$ bits/s, the overall performance of {\em Elf} is supposed to be better than that of Chimp$_{128}$. \textcolor{black}{By adopting the same approach, we can draw a conclusion that the overall performance of {\em Elf}+ is supposed to be better than that of Chimp$_{128}$ if the transmission rate is smaller than $1.96 \times 10^8$ bits/s.}

We want to emphasize two points here. First, in a typical client-server architecture, the bandwidth and memory in the server are rather precious resources, and the bandwidth for a connection rarely exceeds $6.17 \times 10^7$ bits/s \textcolor{black}{(let alone $1.96 \times 10^8$ bits/s)}. Moreover, for each connection, Chimp$_{128}$ would allocate 33KB memory, which is unaffordable for high concurrency scenarios. Second, we find that the most time-consuming part of {\em Elf} \textcolor{black}{or {\em Elf}+} is to calculate $\beta$ \textcolor{black}{or the start position} of a floating-point value. If we could calculate \textcolor{black}{them} faster, the efficiency would be further enhanced tremendously. Maybe in the future we can design a special hardware or a special computer instruction to achieve this.

\subsection{Performance with Different $\beta$  \textcolor{black}{for Double Values}}

\begin{figure}[t]
\centering
{%
	\subfigure[Compression Ratio in AS.]{%
		\includegraphics[width=1.6in]{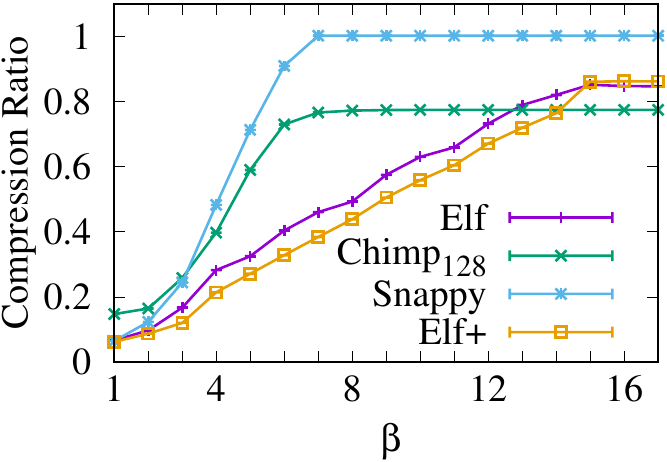}%
	}%
	\hfill%
	\subfigure[Compression Ratio in PLon.]{%
		\includegraphics[width=1.6in]{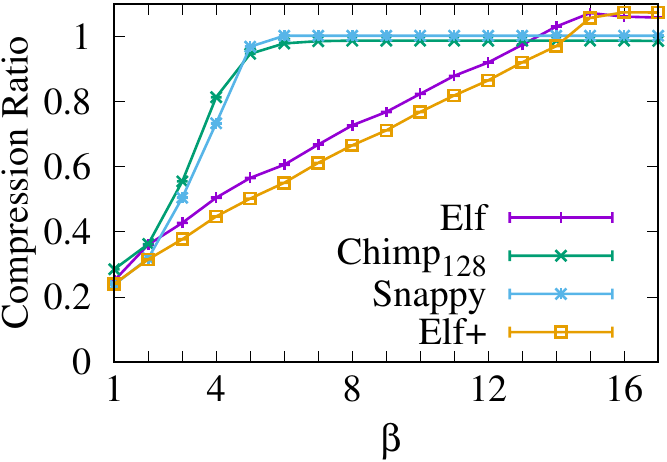}%
	}%

	\subfigure[Compression Time in AS.]{%
		\includegraphics[width=1.6in]{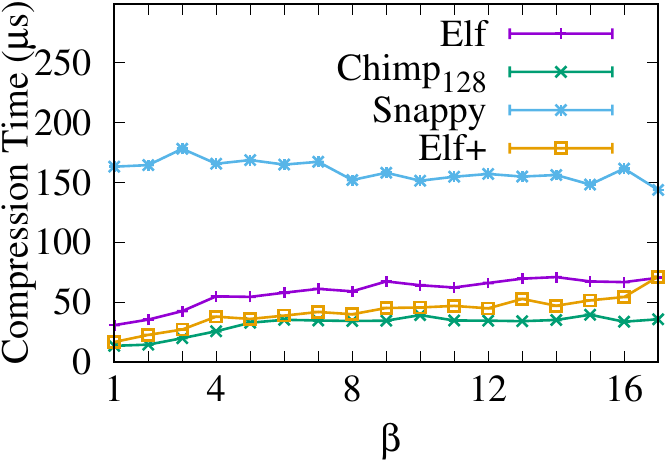}%
	}%
	\hfill%
	\subfigure[Compression Time in PLon.]{%
		\includegraphics[width=1.6in]{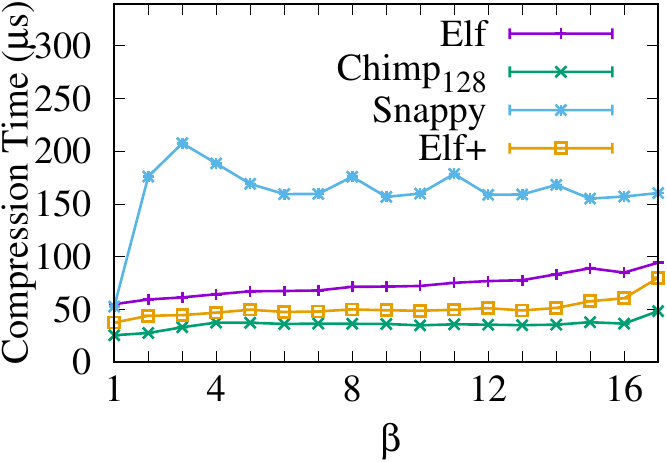}%
	}%
		
	\subfigure[Decompression Time in AS.]{%
		\includegraphics[width=1.6in]{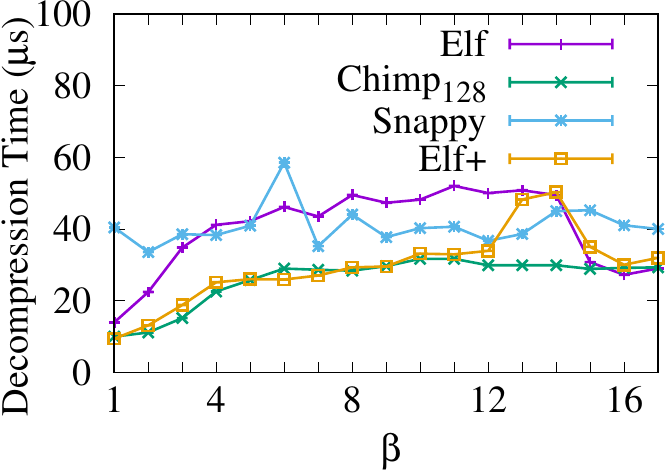}%
	}%
	\hfill%
	\subfigure[Decompression Time in PLon.]{%
		\includegraphics[width=1.6in]{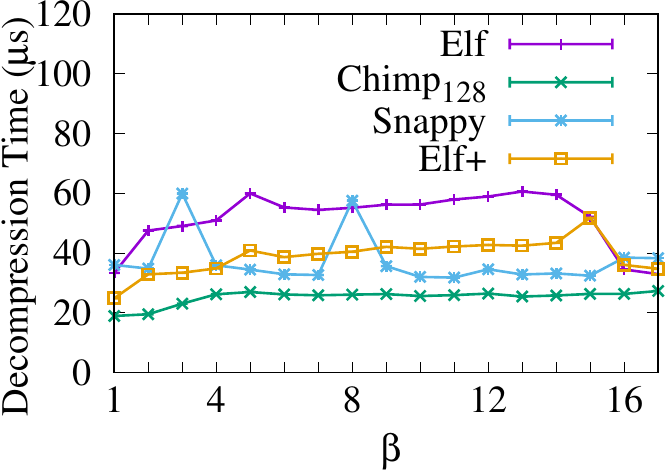}%
	}%
}%
\caption{Performance with Different $\beta$ \textcolor{black}{(Double Values)}.} %
\label{fig:exp:differentbeta}%
\end{figure}%

To further investigate the effect of $\beta$, we conduct a set of experiments by gradually reducing the decimal significand counts of a time series dataset AS and a non time series dataset PLon. We select Chimp$_{128}$ and Snappy as baselines, since they achieve the best trade-off between the compression ratio and compression time among the floating-point competitors and general competitors respectively. 

As shown in Figure~\ref{fig:exp:differentbeta}(a) and Figure~\ref{fig:exp:differentbeta}(b), with an increasing $\beta$ from 1 to 15, the compression ratio of {\em Elf} increases linearly, which is consistent with Theorem~\ref{theorem:effectiveness}. When $\beta$ is greater than 15, the compression ratio of {\em Elf} keeps stable, because {\em Elf} does not perform the erasing step if $\beta > 15$. For Chimp$_{128}$ and Snappy, with the increase of $\beta$, their compression ratios first increase steeply and then keep stable when $\beta > 6$. On both AS and PLon, {\em Elf} always has the best compression ratio \textcolor{black}{compared with Chimp$_{128}$ and Snappy} if $\beta$ is between 3 and 13. When $\beta = 6$, the compression ratio gain of {\em Elf} over Chimp$_{128}$ and Snappy achieves the highest (33\% and 55\% relative improvement in AS, and 40.2\% and 41.6\% relative improvement in PLon, respectively). For the time series dataset AS, {\em Elf} always performs better than Snappy, because {\em Elf} can capture the time ordering characteristic. \textcolor{black}{{\em Elf}+ has a similar compression ratio trend to {\em Elf}. When $\beta < 15$, {\em Elf}+ always performs better than {\em Elf} on both datasets. When $\beta \geq 15$, {\em Elf}+ performs slightly worse than {\em Elf}, as {\em Elf}+ utilizes two bits to indicate the case of not erasing, while {\em Elf} uses only one bit for this case.}

Figures~\ref{fig:exp:differentbeta}(c-f) present the compression time and decompression time of the \textcolor{black}{four} algorithms on the two datasets, respectively. With a larger $\beta$ that $\beta < 15$, the compression time and decompression time of both {\em Elf} and Chimp$_{128}$ get larger, because they need to write or read more streams. Things have changed for Snappy because it contains a complex dictionary building step. When $\beta \geq 15$, the decompression time of {\em Elf} drops sharply, because it skips the restoring step. On both datasets, {\em Elf} takes slightly more compression time than Chimp$_{128}$, but much less than Snappy. Besides, although {\em Elf} takes about double decompression time of Chimp$_{128}$, it is still less than $60 \mu$s for all values of $\beta$. \textcolor{black}{ {\em Elf}+ shows similar trends to {\em Elf} in terms both of compression time and decompression time, but it takes less time for almost all values of $\beta$.}

\begin{figure}[t]
\centering
{%
	\subfigure[Time Series (Small $\beta$).]{%
		\includegraphics[width=1.6in]{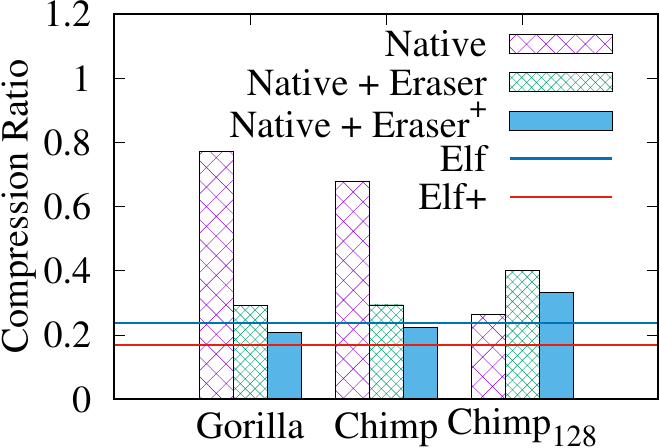}%
	}%
	\hfill%
	\subfigure[Non Time Series (Small $\beta$).]{%
		\includegraphics[width=1.6in]{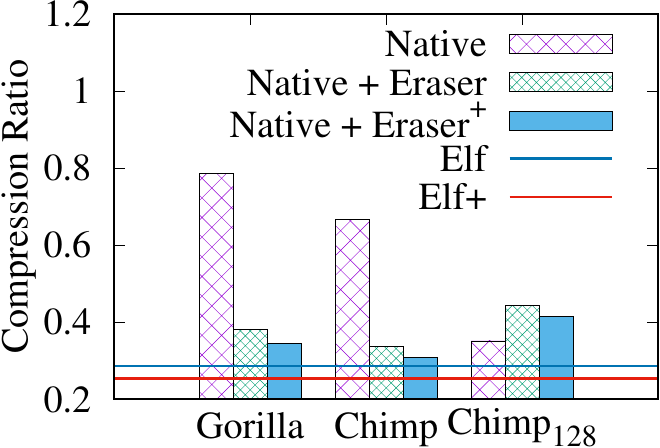}%
	}%
		
	\subfigure[Time Series (Medium $\beta$).]{%
		\includegraphics[width=1.6in]{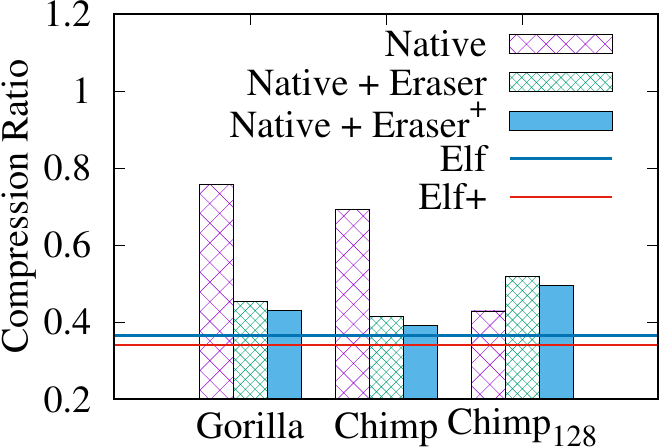}%
	}%
	\hfill%
	\subfigure[Non Time Series (Medium $\beta$).]{%
		\includegraphics[width=1.6in]{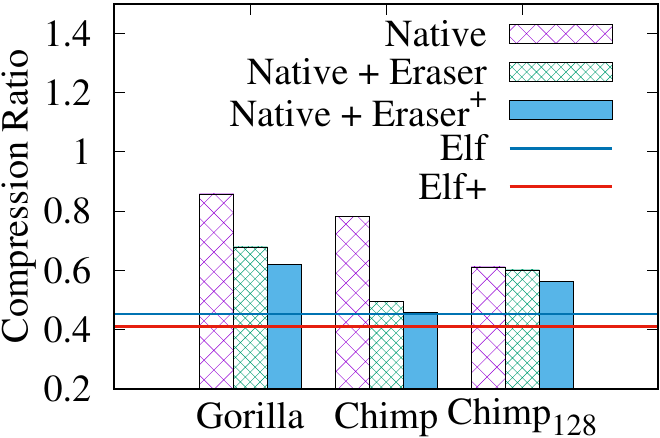}%
	}%
	
	\subfigure[Time Series (Large $\beta$).]{%
		\includegraphics[width=1.6in]{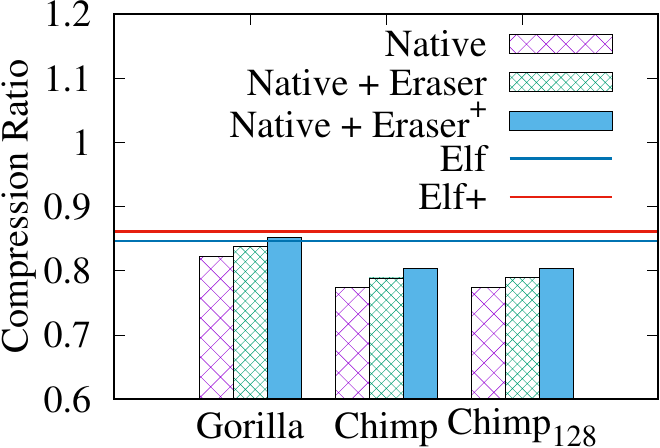}%
	}%
	\hfill%
	\subfigure[Non Time Series (Large $\beta$).]{%
		\includegraphics[width=1.6in]{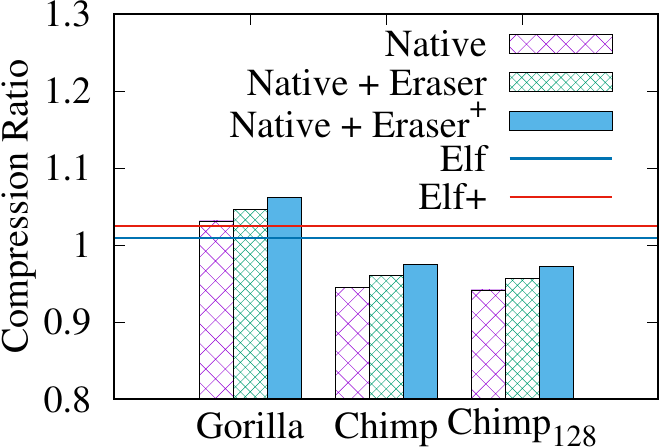}%
	}%
}%
\caption{Compression Ratio Improvement of Erasing and XOR$_{cmp}$ Strategies \textcolor{black}{(Double Values)}.} %
\label{fig:exp:CRimprovement}%
\end{figure}%

\subsection{Validation of Erasing and XOR$_{cmp}$ Strategies}

To verify the effectiveness of the erasing strategy, we regard {\em Elf} Eraser \textcolor{black}{(or {\em Elf}+ Eraser)} as a preprocessing operation on Gorilla, Chimp and Chimp$_{128}$. Figures~\ref{fig:exp:CRimprovement}(a-f) present the average compression ratio improvement over the native methods in three groups of $\beta$. It is observed that: 

(1)~For both time series datasets and non time series datasets with small or medium $\beta$, \textcolor{black}{both of our proposed erasing strategies} can improve the compression ratio of Gorilla and Chimp dramatically. In particular, if $\beta$ is small, with the equipment of {\em Elf} Eraser \textcolor{black}{(or {\em Elf}+ Eraser)}, Gorilla can obtain a relative improvement of 62.2\% and 51.6\% \textcolor{black}{(or 73\% and 56.1\%)} on the time series datasets and non-time series datasets, respectively, while Chimp can also enjoy a relative improvement of 56.8\% and 49.5\% \textcolor{black}{(or 66.9\% and 53.8\%)}, respectively.

(2)~Chimp$_{128}$ can be hardly enhanced by {\em Elf} Eraser \textcolor{black}{and {\em Elf}+ Eraser}. This is because Chimp$_{128}$ leverages the least 14 significant mantissa bits as its hash key. After erasing the mantissa, it is hard for Chimp$_{128}$ to find an appropriate previous value, which might result in an XORed value with a small number of leading zeros. Besides, keeping track of the positions of the chosen values consumes additional bits. As a result, unlike Chimp$_{128}$, {\em Elf} \textcolor{black}{and {\em Elf}+} consider only the neighboring values.

(3)~For datasets with large $\beta$, {\em Elf} Eraser \textcolor{black}{and {\em Elf}+ Eraser} cannot enhance the XOR-based compressors, because for large $\beta$, {\em Elf} Eraser \textcolor{black}{and {\em Elf}+ Eraser} give up erasing to avoid a negative gain.

(4)~If $\beta$ is not large, {\em Elf} \textcolor{black}{(or {\em Elf}+)} is still 8.7\%$\sim$33.3\% \textcolor{black}{(or 10.3\%$\sim$49.3\%)} better than the Eraser-enhanced \textcolor{black}{(or Eraser$^+$-enhanced)} Gorilla and Chimp, which verifies the effectiveness of the optimization for XOR$_{cmp}$.

\subsection{\textcolor{black}{Performance for Single Values}}

\textcolor{black}{We also conduct a set of experiments to verify the performance of the proposed algorithms on single values. For this set of experiments, we use only the datasets with $\beta \leq 7$, since the significand count of a single value would not be greater than 7. FPC does not provide a version of single values, so we do not compare it.}

\textcolor{black}{As shown in Table~\ref{tbl:expForSingle}, although {\em Elf} has a similar compression ratio with that of the best floating-point competitor Chimp$_{128}$, {\em Elf}+ still enjoys the best compression ratio among all the floating-point compression methods. Specifically, compared with Chimp$_{128}$, {\em Elf}+ achieves an average relative compression ratio improvement of 12.8\% and 5.5\% on time series datasets and non time series datasets respectively. Besides, compared with the general compression algorithms, {\em Elf}+ has a better compression ratio than most of them (i.e., LZ4, Zstd and Snappy) and takes significantly less time than all of them. Moreover, like for double values, {\em Elf}+ outperforms {\em Elf} in terms all of compression ratio, compression time and decompression time for single values.}

\textcolor{black}{It is also observed that the compression ratios of {\em Elf} and {\em Elf}+ for single values are slightly worse than those of them for double values, respectively, but their compression/decompression times are not much different. For example, the average compression ratio of {\em Elf}+ is 0.33 for time series of double values, but it turns into 0.41 for time series of single values. This is because single values take up much fewer mantissa bits than double values, and thus we can only erase fewer bits for single values. In fact, other methods including floating-point specific compression algorithms and general compression algorithms show the same results.}

\setlength{\tabcolsep}{0.33em} 
\begin{table}[t]
\caption{\textcolor{black}{Average performance for single values (the best values in each group are in \textbf{bold}). The compression ratio (CR), compression time (CT) and decompression time (DT) are the average measurements on one block (i.e., 1,000 values).}}
\arrayrulecolor{black}
\centering\label{tbl:expForSingle}
\resizebox{3.33in}{!}{
\begin{tabular}{|c|c||c|c|c||c|c|c|} 
\hline
\multicolumn{2}{|c||}{\multirow{2}{*}{\textbf{Dataset}}}	&\multicolumn{3}{c||}{\textbf{Time Series}}	&\multicolumn{3}{c|}{\textbf{Non Time Series}}\\
\cline{3-8}

\multicolumn{2}{|c||}{}	&CR		&	CT ($\mu$s)	& DT ($\mu$s)	& CR	& CT ($\mu$s)	&DT ($\mu$s)\\

\hline
\hline

\multirow{5}{*}{\rotatebox[origin=c]{90}{\textbf{Floating}}}	& Gorilla	&0.66	&\textbf{18.0}	&\textbf{15.3}	&0.85&\textbf{19.3}	&\textbf{15.8}	\\
& Chimp	&0.57	&19.8	&16.9	&0.78&23.4	&19.0	\\
& Chimp$_{128}$	&0.47	&26.4	&17.6	&0.73&33.3	&20.1	\\
& Elf	&0.46	&56.4	&43.1	&0.74	&63.6	&47.9	\\
& Elf+	&\textbf{0.41}	&41.4	&32.0	&\textbf{0.69}&51.5	&37.1	\\
\hline
\hline

\multirow{5}{*}{\rotatebox[origin=c]{90}{\textbf{General}}}		& Xz	&\textbf{0.36}	&979.5	&175.6	&\textbf{0.60}	&1054.0	&247.2	\\
& Brotli	&0.40	&1660.5	&89.3	&0.63	&1588.8	&80.0	\\
& LZ4	&0.72	&1064.5	&42.6	&0.80	&1004.6	&39.3	\\
& Zstd	&0.44	&229.7	&66.2	&0.65	&226.2	&55.1	\\
& Snappy	&0.69	&\textbf{187.1}	&\textbf{41.9}	&0.83	&\textbf{183.7}	&\textbf{36.4}	\\
\hline
\end{tabular}
}
\end{table}

\section{Related Works}\label{sec:related}


\subsection{General Compression}

There are a wide range of impressive compression methods for general purposes, such as Xz~\cite{Xz}, Brotli~\cite{alakuijala2018brotli}, LZ4~\cite{collet2013lz4}, Zstd~\cite{collet2016zstd} and Snappy~\cite{snappy}. Zstd combines a dictionary-matching stage with a fast entropy-coding stage. The dictionary is trainable and can be generated from a set of samples. Snappy also refers to a dictionary and stores the shift from the current position back to uncompressed stream. Both Zstd and Snappy can achieve a good trade-off between compression ratio and efficiency. Most general compression methods are lossless and can achieve a good compression ratio, but they do not leverage the characteristics of floating-point values and cannot be applied directly to streaming scenarios~\cite{li2020discovering} either.

\subsection{Lossy Floating-Point Compression}

Since floating-point data is stored in a complex format, it is challenging to compress floating-point data without losing any precision. To this end, many lossy floating-point compression methods are proposed~\cite{lazaridis2003capturing, liang2022sz3, lindstrom2014fixed, zhao2022mdz, zhao2021optimizing, liu2021high, liu2021decomposed}. For example, the representative method ZFP~\cite{lindstrom2014fixed} compresses regularly gridded data with a certain loss guarantee. MDZ~\cite{zhao2022mdz} is an adaptive error-bounded lossy compression framework that optimizes the compression for two execution models of molecular dynamics. However, these lossy compression methods are usually application specific. Moreover, many scenarios, especially in the fields of scientific calculation and databases~\cite{xiao2022time, li2020just, Yu2021distributed, bao2016managing}, do not tolerate any loss of precision.

\subsection{Lossless Floating-Point Compression}

Most lossless floating-point compression algorithms are based on prediction. The distinction among them lies in two aspects: 1)~How does the predictor work? 2)~How to handle the difference between the predicted value and the real one? 

Based on the former, lossless floating-point compression algorithms can be further divided into model-based methods~\cite{ratanaworabhan2006fast, yu2020two, burtscher2007high, burtscher2008fpc, jensen2018modelardb, jensen2021scalable, blalock2018sprintz} and previous-value methods~\cite{liakos2022chimp, pelkonen2015gorilla}. DFCM~\cite{ratanaworabhan2006fast} maps floating-point values to unsigned integers and predicts the values by a DFCM (differential finite context method) predictor. However, DFCM only works well for smoothly changing data. FPC~\cite{burtscher2007high, burtscher2008fpc} sequentially predicts each value in a streaming fashion using two context-based predictors, i.e., FCM predictor~\cite{sazeides1997predictability} and DFCM predictor (which is quite different from that in DFCM~\cite{ratanaworabhan2006fast}). Among the predicted values obtained by the two predictors, FPC chooses the closer one, and thus it can achieve a better prediction performance. Some other model-based methods~\cite{jensen2018modelardb, jensen2021scalable, yu2020two} capture the characteristics of different series using machine learning models, and eventually choose the best compression approach. Due to the high cost of prediction, Gorilla~\cite{pelkonen2015gorilla} and Chimp~\cite{liakos2022chimp} directly regard the previous one value as the predicted one, based on the observation that two consecutive values do not change much. Chimp$_{128}$ is an upgraded version of Chimp, which exploits 128 earlier values to find the best matched value. To expedite the computation efficiency, Chimp$_{128}$ maintains a hash table with size of 33KB, which might be not applicable in edge computing scenarios~\cite{shi2016edge, mao2017survey}.

Based on the latter, a small number of methods~\cite{engelson2000lossless} first map the differences between the predicted values and actual values to integers, and then compress the integers using integer-oriented compression techniques such as Delta encoding~\cite{pelkonen2015gorilla}. On the contrary, a majority of methods~\cite{burtscher2008fpc, liakos2022chimp, pelkonen2015gorilla} encode their XORed values instead of the differences. Gorilla~\cite{pelkonen2015gorilla} assumes that the XORed values would contain both long leading zeros and long trailing zeros with high probability, so it uses 5 bits to record the number of leading zeros and 6 bits to store the number of trailing zeros. Chimp~\cite{liakos2022chimp} points out the fact that the XORed values rarely have long trailing zeros, so it is ineffective for Gorilla to take up to 6 bits to record the number of trailing zeros. Therefore, Chimp optimizes the encoding strategy for the XORed values and can use fewer bits. 

As a lossless compression solution, {\em Elf} belongs to a previous-value method and encodes the XORed values. However, different from Gorilla and Chimp, {\em Elf} performs an erasing operation on the floating-point values before XORing them, which makes the XORed values contain many trailing zeros. Besides, {\em Elf} designs a novel encoding strategy for the XORed values with many trailing zeros, which achieves a notable compression ratio.

\section{Conclusion and Future Work}\label{sec:conclude}

This paper \textcolor{black}{first puts forward} a novel, compact and efficient erasing-based lossless floating-point compression algorithm {\em Elf}, \textcolor{black}{and then proposes an upgraded version of it named {\em Elf}+ by optimizing the significand count encoding strategy}.
Extensive experiments using 22 datasets verify the powerful performance of {\em Elf} \textcolor{black}{and {\em Elf}+ for both double values and single values}. In particular, \textcolor{black}{for double values,} {\em Elf} achieves average relative compression ratio improvement of 12.4\% and 43.9\% over Chimp$_{128}$ and Gorilla, respectively. Besides, {\em Elf} has a similar compression ratio to the best compared general compression algorithm but with much less time. \textcolor{black}{Furthermore, {\em Elf}+ outperforms {\em Elf} by an average relative compression ratio improvement of 7.6\% and compression time improvement of 20.5\%.}
In our future work, we plan to optimize {\em Elf} for specific data types, such as trajectories.

\begin{acknowledgements}
This work was supported by the National Natural Science Foundation of China (62202070, 61976\\168, 62172066, 62076191) and China Postdoctoral Science Foundation (2022M720567).
\end{acknowledgements}

\bibliographystyle{spmpsci}      
\bibliography{ref}   

\end{document}